\documentclass[journal]{IEEEtran}
\usepackage{amsmath,amssymb,amsfonts,bm}
\usepackage{amsthm}
\usepackage{graphicx}
\usepackage{mathtools}
\usepackage{mathrsfs}
\usepackage{tabularx}
\usepackage{bbm}
\usepackage{cite}
\usepackage{booktabs}
\usepackage{array}
\usepackage{multirow}
\usepackage{xcolor}
\usepackage{textcomp}
\usepackage{url}
\usepackage{algorithm}
\usepackage{algpseudocode}
\allowdisplaybreaks

\newcommand{\EnvMap}{\Omega}
\newcommand{\SeviceMap}{\Omega_\text{s}}

\newcommand{\SwarmState}{S}
\newcommand{\UAVPos}{\mathbf p}

\newcommand{\PrimAct}{a}
\newcommand{\ActSet}{\mathbb A}
\newcommand{\RSRP}{R}

\newcommand{\UAVSet}{\mathbb N}
\newcommand{\CommittedUAVSet}{\mathbb N^\mathrm{com}}
\newcommand{\ActiveUAVSet}{\mathbb N^\mathrm{act}}
\newcommand{\NumUAV}{N}

\newcommand{\TransKernel}{P}
\newcommand{\SyncPeriod}{H}
\newcommand{\Msg}{m} 

\newcommand{\EdgeSet}{\mathbb E_l} 

\newcommand{\UAVNeighbor}{\mathbb N^{\mathrm{nbr}}}

\newcommand{\Gain}{G}

\newcommand{\RoundLimit}{R}
\newcommand{\Traj}{\tau} 

\newcommand{\CommittedTraj}{\Traj^\mathrm{com}} 

\newcommand{\CounterfactualTraj}{\Traj^\mathrm{CF}} 

\newcommand{\TrajSet}{\mathbb T} 
\newcommand{\CommittedTrajSet}[1]{\TrajSet^\mathrm{com}_{<#1}} 

\newcommand{\FeasProposalPool}{\TrajSet^\mathrm{feas}}
\newcommand{\LocalProposalPool}{\TrajSet^\mathrm{loc}}
\newcommand{\SubgraphProposalPool}{\TrajSet^\mathrm{sub}}
\newcommand{\GlobalProposalPool}{\TrajSet^\mathrm{global}}

\newcommand{\GraphSet}{\mathbb H_t}
\newcommand{\SubgraphSet}{\UAVSet_l}

\newcommand{\ProposalScore}{\upsilon}
\newcommand{\ElectionChoice}{\chi}

\newcommand{\SelectedUAV}{i^\star}

\newcommand{\NumFeas}{K_\text{feas}} 
\newcommand{\NumCand}{K_\text{c}} 
\newcommand{\RelayDepth}{\mathcal K_\text{dep}} 

\newcommand{\WorldModel}{\mathcal W_{\WorldParams}}
\newcommand{\ElectionModule}{\mathcal A_{\ElectionParams}}

\newcommand{\DeployEncoder}{\mathcal E_{\WorldParams}}
\newcommand{\DeployPredictor}{\mathcal P_{\WorldParams}}
\newcommand{\TargetEncoder}{\mathcal E_{\TargetParams}}
\newcommand{\TargetPredictor}{\mathcal P_{\TargetParams}}

\newcommand{\WorldParams}{\theta}
\newcommand{\TargetParams}{\xi}
\newcommand{\ElectionParams}{\phi}
\newcommand{\CriticParams}{\nu}

\newcommand{\JEPALoss}{\mathcal L_{\mathrm{JEPA}}}
\newcommand{\TargetLoss}{\mathcal L_{\mathrm{target}}}

\newcommand{\TargetReconLoss}{\mathcal L_{\mathrm{rec}}}

\newcommand{\RankLoss}{\mathcal L_{\mathrm{rank}}}
\newcommand{\AlignLoss}{\mathcal L_{\mathrm{align}}}

\newcommand{\CoherenceLoss}{\mathcal L_{\mathrm{coh}}}

\newcommand{\ElectionUpdateLoss}{\mathcal L_{\mathrm{Election}}}

\newcommand{\HuberLoss}{\ell_{\mathrm H}}

\newcommand{\CoherenceTemp}{\tau_{\mathrm{coh}}}

\newcommand{\EndpointDist}[2]{d_{\mathrm{geo}}(#1,#2)}
\newcommand{\LatentDist}[2]{d_{\mathrm z}(#1,#2)}
\newcommand{\LatentEps}{\varepsilon_{\mathrm z}}
\newcommand{\LatentNorm}[1]{\operatorname{norm}(#1)}

\newcommand{\ProposalEndpoint}{\UAVPos}

\newcommand{\ElectionReward}{r}
\newcommand{\ConnPenalty}{\lambda_{\mathrm c}}

\newcommand{\PPOClipLoss}{\mathcal L_{\mathrm{clip}}}
\newcommand{\PPOValueLoss}{\mathcal L_{\mathrm{value}}}
\newcommand{\PPOEntropyBonus}{\mathcal H_{\mathrm{EA}}}
\newcommand{\PPOValueCoef}{c_{\mathrm v}}
\newcommand{\PPOEntropyCoef}{c_{\mathrm e}}
\newcommand{\PPOClipEps}{\epsilon_{\mathrm{clip}}}
\newcommand{\PPOProbRatio}{\rho}
\newcommand{\PPOAdvantage}{A}
\newcommand{\PPOValue}{V}
\newcommand{\PPOReturn}{G}

\newcommand{\FreshRollout}{\mathcal R^{\mathrm{new}}}

\newcommand{\WarmupSteps}{N_{\mathrm{warm}}}
\newcommand{\LearnSteps}{N_{\mathrm{learn}}}

\newcommand{\JointPlan}{\TrajSet}
\newcommand{\BasePlan}{\TrajSet^{\mathrm{base}}}
\newcommand{\SeqPlan}{\TrajSet^\mathrm{seq}}
\newcommand{\ParallelPlan}{\TrajSet^{\parallel}}

\newcommand{\DeployInfo}{\mathcal I}

\newcommand{\WarmupAct}{a^{\mathrm{warm}}} 

\newcommand{\DispatchSep}{\psi^{\mathrm{sep}}}
\newcommand{\DispatchCrowd}{\psi^{\mathrm{crowd}}}
\newcommand{\DispatchDegrade}{\psi^{\mathrm{deg}}}
\newcommand{\DispatchSepWeight}{\lambda_{\mathrm{sep}}}
\newcommand{\DispatchCrowdWeight}{\lambda_{\mathrm{crowd}}}
\newcommand{\DispatchDegradeWeight}{\lambda_{\mathrm{deg}}}
\newcommand{\LocalObs}{o^{\mathrm{loc}}}

\newcommand{\Bandwidth}{B}
\newcommand{\Power}{P}
\newcommand{\RefSignalPower}[1]{\Power_{#1}^{\mathrm{rs}}}
\newcommand{\GroundPower}{\Power}
\newcommand{\SINR}{\gamma}
\newcommand{\Rate}{r}
\newcommand{\Interference}{I}
\newcommand{\RadioField}{Q}
\newcommand{\CoverageThresh}{\rho_{\mathrm{cov}}}
\newcommand{\SINRThresh}{\phi_{\mathrm{SINR}}}
\newcommand{\LinkThresh}{\rho_{\mathrm{link}}}
\newcommand{\BackhaulThresh}{\rho_{\mathrm{bh}}}

\newcommand{\RootSet}{\mathbb N^{\mathrm{root}}}

\newcommand{\Connectivity}{C}
\newcommand{\ConnIndicator}{b}
\newcommand{\MetricVec}{\bm{\mu}}
\newcommand{\MetricWeight}{\mathbf w}
\newcommand{\BalancedScore}{J}

\newcommand{\CFGain}{\Delta^{\mathrm{CF}}}
\newcommand{\CompositionGap}{\mathcal G^{\mathrm{CF}}}

\newcommand{\MainNotationTable}{%
\begin{table}[!t]
\caption{Main notations used in the methodology.}
\label{tab:notation}
\centering
\setlength{\tabcolsep}{3pt}
\renewcommand{\arraystretch}{1.08}
\begin{tabularx}{\columnwidth}{@{}>{\raggedright\arraybackslash}m{0.35\columnwidth}X@{}}
\toprule
Symbol & Meaning \\
\midrule
\multicolumn{2}{@{}l}{\textit{System and deployment information}}\\
\vspace{2pt}
$\UAVSet,\NumUAV$
& UAV set and number of UAVs. \\
$\UAVPos_{i,t}$
& Position of UAV $i$ at deployment step $t$. \\
$\GraphSet,\SubgraphSet$
& A2A communication graph and one subgraph. \\
$\UAVNeighbor_i$
& One-hop neighbor set of UAV $i$. \\
$\LocalObs_i$
& Deployment-visible local observation of UAV $i$. \\
$\RadioField_t$
& Radio field induced by the UAV configuration at step $t$. \\
$\MetricVec,\MetricWeight,\BalancedScore$
& Objective metric vector, metric weights, and balanced objective score. \\
$\JointPlan_t$
& Joint plan executed after synchronization step $t$. \\
\midrule
\multicolumn{2}{@{}l}{\textit{Proposal election}}\\
\vspace{2pt}
$\SyncPeriod,\RoundLimit$
& Synchronization horizon and maximum number of negotiation rounds. \\
$\RelayDepth$
& Maximum relay depth for inter-subgraph proposal exchange. \\
$\NumFeas,\NumCand$
& Number of feasible trajectories and retained candidate proposals. \\
$\Traj_{i,k},\FeasProposalPool$
& Candidate trajectory of UAV $i$ and feasible trajectory set. \\
$\Msg_{j\to i}^{r}$
& Message from UAV $j$ to UAV $i$ at negotiation round $r$; it contains a set of proposals. \\
$\SubgraphProposalPool_r$
& Proposal pool after deterministic reduction inside a subgraph. \\
$\GlobalProposalPool_r$
& Proposal pool after relay-bounded inter-subgraph aggregation. \\
$\ElectionChoice_r$
& Election output at round $r$. \\
$\CommittedTraj_r,\CommittedTrajSet{r}$
& Trajectory committed at round $r$ and context of trajectories committed before round $r$. \\
$\CommittedUAVSet_r,\ActiveUAVSet_r$
& UAVs already committed by round $r$ and UAVs still active at round $r$. \\
\midrule
\multicolumn{2}{@{}l}{\textit{JEPA radio world model and training}}\\
\vspace{2pt}
$\WorldModel,\ElectionModule$
& JEPA radio world model and election module. \\
$\TargetEncoder,\DeployEncoder$
& Target encoder trained from rollout and deployment encoder used before election. \\
$\Delta\MetricVec_{i,k}^{r}$
& Rollout-measured objective metric change for trajectory $\Traj_{i,k}$ at round $r$. \\
$z_{i,k}^{+,r},\hat z_{i,k}^{r}$
& Target latent and deployment-predicted latent. \\
$\ProposalScore_{i,k}$
& Scalar proposal score used by ranking and reduction. \\
$\JEPALoss,\TargetLoss,\AlignLoss$ & JEPA objective, target-side loss,and alignment loss. \\
$\CoherenceLoss,\RankLoss$ & Coherence regularization and auxiliary ranking loss. \\
$\WorldParams,\TargetParams,\ElectionParams,\CriticParams$
& Parameters of the world model, target path, election module, and critic. \\
$\FreshRollout$
& Fresh negotiation rollout collected after a representation update. \\
\bottomrule
\end{tabularx}
\vspace{-8pt}
\end{table}%
}

\theoremstyle{definition}
\newtheorem{definition}{Definition}
\newtheorem{assumption}{Assumption}
\theoremstyle{plain}

\newtheorem{theorem}{Theorem}

\theoremstyle{remark}
\newtheorem*{remark}{Remark}

\title{WONDER: A Radio World Model-based Negotiation Framework for Multi-Agent UAV Coverage Optimization}
\author{
\IEEEauthorblockN{Jiahao Huang, 
Rongpeng Li, 
Zhifeng Zhao, Guoru Ding, 
and Honggang Zhang
\thanks{J. Huang and R. Li are with Zhejiang University, Hangzhou 310027, China,
(email: \{22331083, lirongpeng\}@zju.edu.cn).}\thanks
{Z. Zhao is with Zhejiang Lab, Hangzhou 310012, China, as well as Zhejiang University,
 Hangzhou 310027, China (email: zhaozf@zhejianglab.org).}\thanks{G. Ding is with College of Communications and Engineering, Army Engineering University of PLA, Nanjing 210007, China (e-mail: dr.guoru.ding@ieee.org).
}
\thanks{H. Zhang is with Macau University of Science and Technology, Macau, China (email: hgzhang@must.edu.mo).}
 }
}
\begin{document}
\maketitle

\begin{abstract}
Post-disaster damage to terrestrial infrastructure can disrupt wireless coverage,
while Uncrewed Aerial Vehicle (UAV) swarms provide a promising solution for rapid restoration.
However, due to the limitations in local geometry observations
hidden radio impact,
and inter-UAV communication,
there exists a significant gap between locally visible movement choices and swarm-level coverage outcomes.
To combat this gap,
we propose a raido \underline{W}orld-model-based \underline{O}ptimized \underline{N}egotiation framework for \underline{D}istributed UAV cov\underline{ER}age (WONDER).
Particularly, to tackle the unavailability of the future radio field from onboard observations, 
WONDER uses a Joint-Embedding Predictive Architecture (JEPA)-based radio world model to 
learn and predict the incremental radio effect of each candidate trajectory from deployment-available information.
Multi-round negotiation in WONDER then coordinates ranked proposals by committing one trajectory at a time and re-evaluating the remaining proposals under the updated context. Our theoretical analyses further validate the effectiveness of such a world model-based framework.
WONDER also adopts a Proximal Policy Optimization (PPO)-style Actor and alternates between updating the world model and the actor.
Furthermore,
we build RadioDynamics,
a comprehensive simulation environment that integrates UAV mobility,
radio propagation, inter-UAV communication modeling,
and digital-twin geometry with ray-traced fields in $62$ metropolitan scenes.
Experiments on $11$ testing scenes in RadioDynamics show that WONDER achieves the highest balanced score among seven evaluated methods,
reaching $0.870$ with a $0.162$ coverage advantage over STACCA, while maintaining $100\%$ connectivity between UAVs.
\end{abstract}


\section{Introduction}
\label{sec:introduction}

Post-disaster damage to terrestrial communication infrastructure can disrupt wireless
services across affected areas,
creating an urgent need for rapid coverage restoration~\cite{kwasinski2009telecommunications}.
Correspondingly, Uncrewed Aerial Vehicle (UAV) swarms offer a flexible and resilient solution~\cite{mozaffari2019tutorial,zeng2019accessing,viet2022aerial}, 
by establishing adaptive connectivity graphs to provide communication services for end users~\cite{xu2026scalable,romero2024aerial}. 
Although recent studies have advanced radio-map-assisted UAV deployment~\cite{li2024radio},
learning-based multi-UAV coverage and trajectory control~\cite{wang2026robust},
and coverage--connectivity optimization~\cite{cai2025multiuav},
local yet partial environmental observations
and limited inter-UAV communication could compromise the effectiveness of UAV deployment decisions.
Therefore, a key deployment challenge remains: expanding spatial coverage without breaking inter-UAV connectivity when radio feedback is unavailable.

\subsection{Related Works}
Deployable multi-UAV coverage optimization with maintained inter-UAV connectivity builds on three lines of research.
UAV-enabled wireless coverage and emergency networking examine how mobile aerial access and relay nodes support service restoration.
Multi-Agent Reinforcement Learning (MARL) studies cooperative control with partial observations and limited message exchange.
Radio world models relate propagation,
geometry,
and radio measurements to enhance coverage decisions under unavailable radio feedback.

\begin{table*}[!t]
\centering
\caption{A key summary of the key differences between WONDER and the literature.}
\label{tab:positioning}
\setlength{\tabcolsep}{1.1pt}
\renewcommand{\arraystretch}{1.08}
\newcommand{\tabCircleMark}[1]{%
\raisebox{-0.14ex}{%
\begingroup
\setlength{\unitlength}{0.62em}%
\begin{picture}(1,1)
#1%
\end{picture}%
\endgroup}}
\newcommand{\tabYes}{\tabCircleMark{\put(0.5,0.5){\circle*{1}}}}
\newcommand{\tabNo}{\tabCircleMark{\put(0.5,0.5){\circle{1}}}}
\newcommand{\tabPart}{%
\tabCircleMark{%
\put(0.5,0.5){\circle*{1}}
{\color{white}\put(0.5,0){\rule{0.5\unitlength}{1\unitlength}}}
\put(0.5,0.5){\circle{1}}
}}
\newcommand{\tabRowRule}{\specialrule{0.25pt}{0.6pt}{0.6pt}}
\begin{tabularx}{\textwidth}{@{}>{\raggedright\arraybackslash}m{0.23\textwidth}*{6}{>{\centering\arraybackslash}X}@{}}
\toprule
Representative Studies
& \shortstack{UAV Swarm\\Coverage}
& \shortstack{Radio-Aware\\Evaluation}
& \shortstack{Predictive\\World Model}
& \shortstack{Sparse\\Communication}
& \shortstack{Long-Range\\Coordination}
& \shortstack{Trajectory\\Negotiation} \\
\midrule
UAV access and emergency relay~\cite{zeng2019accessing,mozaffari2019tutorial,viet2022aerial,shamsoshoara2021unvail,tran2022uavrelay,xu2026scalable}
& \tabYes & \tabPart & \tabNo & \tabPart & \tabPart & \tabNo \\
\tabRowRule
Multi-UAV coverage control~\cite{chen2024transformer,cai2025multiuav,xiao2024deep,carvalho2026multi}
& \tabYes & \tabPart & \tabNo & \tabNo & \tabNo & \tabNo \\
\tabRowRule
Radio-aware deployment and radio maps~\cite{romero2024aerial,romero2022radio,li2026lanc}
& \tabPart & \tabYes & \tabPart & \tabNo & \tabNo & \tabNo \\
\tabRowRule
World models and action-sufficient representations~\cite{hafner2023dreamerv3,zhang2025dima,balestriero2025lejepa,assran2025vjepa2,martinez2026coral,hyeon2026asgr}
& \tabNo & \tabNo & \tabYes & \tabPart & \tabNo & \tabNo \\
\tabRowRule
Networked MARL and long-range critics~\cite{yu2022surprising,dgn2020,wen2022mat,sinha2025stacca}
& \tabNo & \tabNo & \tabNo & \tabPart & \tabYes & \tabNo \\
\tabRowRule
Learned communication~\cite{das2019tarmac,ding2024magi,li2024cacom,hu2023etcnet,code2025}
& \tabNo & \tabNo & \tabNo & \tabYes & \tabPart & \tabNo \\
\tabRowRule
\textbf{WONDER (Ours)}
& \textbf{\tabYes} & \textbf{\tabYes} & \textbf{\tabYes} & \textbf{\tabYes} & \textbf{\tabYes} & \textbf{\tabYes} \\
\midrule
\multicolumn{7}{r@{}}{\scriptsize \tabYes~Yes \quad \tabPart~Partial \quad \tabNo~No} \\
\end{tabularx}
\vspace{-16pt}
\end{table*}

\subsubsection{UAV Swarm Coverage and Emergency Networking}

UAV-mounted aerial base stations provide flexible access and relay services when terrestrial wireless infrastructure is unavailable or impaired~\cite{zeng2019accessing,mozaffari2019tutorial,viet2022aerial}.
Extending from individual placement to joint coverage optimization, 
recent studies have examined scalable area coverage~\cite{chen2024transformer},
coverage--connectivity tradeoffs for user equipment~\cite{cai2025multiuav},
dynamic coverage planning in complex environments~\cite{xiao2024deep},
and policy transfer across dynamic scenarios~\cite{carvalho2026multi}.
%
Nevertheless, post-disaster UAV networking couples aerial service provision,
connectivity maintenance,
and mobility control~\cite{shamsoshoara2021unvail,tran2022uavrelay,xu2026scalable}.
In this paper, we focus on coverage restoration under deployable information,
where UAVs make coverage optimization decisions from local observations and limited inter-UAV communication without available radio feedback and compromising inter-UAV connectivity.
%

\subsubsection{Multi-Agent Reinforcement Learning for Networked Coordination}

MARL, which typically leverages centralized training with decentralized execution~\cite{yu2022surprising} to address a partially observable decision-making problem, sounds natural for coverage restoration. Furthermore, the decisions of UAVs are coupled, and the eventual impact of an individual action by one UAV depends on the concurrent choices of its counterparts. 
Therefore, 
graph-based information aggregation~\cite{dgn2020}
and sequential joint-action modeling~\cite{wen2022mat}
%
could be essential. Specifically, inter-agent communication content \cite{das2019tarmac} can be calibrated according to task relevance \cite{ding2024magi},
receiver context \cite{li2024cacom}, timing \cite{code2025}, and network constraints \cite{hu2023etcnet}. On the other hand, to support credit assignment for local decisions, a counterfactual critic, which uses long-range interaction context, has recently been used \cite{sinha2025stacca}. Unfortunately, as shown lately, we find this creates a significant coordination gap between long-range counterfactual evaluation and synchronized parallel joint execution,
motivating a rethinking of how to more meaningfully utilize contextual information to derive coverage optimization-oriented trajectory proposals.

\subsubsection{Radio-Aware Deployment and Radio World Model}

Radio-aware UAV deployment links mobility decisions with propagation-dependent service quality.
Prior studies have used radio maps,
environmental geometry,
and coverage prediction to support aerial base-station placement and low-altitude network planning~\cite{romero2024aerial,romero2022radio,li2026lanc}.
These works show that radio information is essential for coverage-aware mobility,
especially when building blockage and spatial propagation variation affect service quality.
In coverage restoration scenarios, where no radio feedback is available,
predictive representations that relate locally visible geometry and candidate movements to radio impact could benefit subsequent planning and control from partial observations~\cite{hafner2023dreamerv3,zhang2025dima}.
Joint-Embedding
Predictive Architecture (JEPA) methods extend this idea by learning latent representations of future or hidden information,
without explicit reconstruction~\cite{balestriero2025lejepa,assran2025vjepa2}. In other words,
 a learned representation matters more with downstream decisions,
by preserving essential distinctions needed for action selection~\cite{martinez2026coral,hyeon2026asgr}.
Hence, it is meaningful to investigate how to complement MARL with a radio world model that predicts decision-relevant radio impact from deployment-visible information.

\subsection{Motivation and Contributions}
Post-disaster coverage restoration requires relating visible geometry,
candidate UAV movements,
and inter-UAV messages to the coverage and service-quality effects that
are not directly observed at decision time. Correspondingly, it encounters the following two prominent challenges:
\begin{itemize}
    \item How can UAVs infer decision-relevant radio impact from deployment-visible
        geometry
        when each UAV observes only a local part of the scene and exchanges only
        few messages with nearby UAVs?

    \item How can UAVs coordinate and optimize swarm-level trajectories 
        while reducing the gap between counterfactual evaluations and synchronized
        joint execution under intermittent global synchronization~\cite{lauri2017periodic,kantaros2016intermittent}?
\end{itemize}

To answer these questions,
we propose a World-model-based Optimized Negotiation framework for Distributed UAV covERage (WONDER),
a task-oriented proposal negotiation framework for decentralized UAV-swarm coverage.
WONDER forms candidate trajectory proposals and predicts their radio
impact using a radio world model, learned through a JEPA, from deployment-available
information. 
Through multi-round negotiation,
an election actor coordinates ranked trajectory proposals under limited inter-UAV
communication and re-evaluates the remaining proposals after each negotiated
commitment.
Besides summarizing the key differences with the literature in Table \ref{tab:positioning}, the main contributions of this paper are summarized as follows:

\begin{itemize}
    \item We develop WONDER,
          a task-oriented proposal negotiation framework for decentralized UAV-swarm coverage optimization.
          WONDER uses a JEPA-based radio world model to capably predict each candidate trajectory's incremental radio impact under the evolving negotiation context.
          Thus, it enables a Proximal Policy Optimization (PPO)-based trajectory decision even when direct radio measurement is unavailable.
          Furthermore, through alternate updates between the world model and PPO actor, WONDER ensures the PPO actor is trained on proposal pools induced by the latest world model, 
          reducing the mismatch between proposal generation and election learning.

    \item For WONDER, we provide an empirical and theoretical analysis of multi-round negotiation. 
        The analysis identifies the counterfactual composition gap produced by synchronized parallel execution of individually evaluated proposals. 
        It shows that the induced sequential plan instead admits a recursive lower bound over the parallel counterfactual plan. 
        With proposal coverage and bounded election regret, this bound becomes attainable by WONDER through its generated proposal pool and PPO-based election actor.

    \item As shown in Fig. \ref{fig:real}, towards evaluating deployable UAV-swarm coverage optimization, 
        we build RadioDynamics, a simulation environment that integrates UAV mobility, radio propagation and inter-UAV communication modeling, and digital-twin geometry with ray-traced fields. 
        RadioDynamics constructs $62$ radio scenes, each corresponding to a \(700\,\mathrm{m}\times700\,\mathrm{m}\) central urban district in metropolises such as Hong Kong, New York, and Tokyo. 
        The simulated radio fields are further grounded in realistic three-dimensional building geometry and urban propagation structure.
        The closed-loop experiments on $11$ testing scenes in RadioDynamics validate the superiority of WONDER over STACCA~\cite{sinha2025stacca}.
\end{itemize}
\begin{figure*}[!t]
\centering
\includegraphics[width=0.95\textwidth]{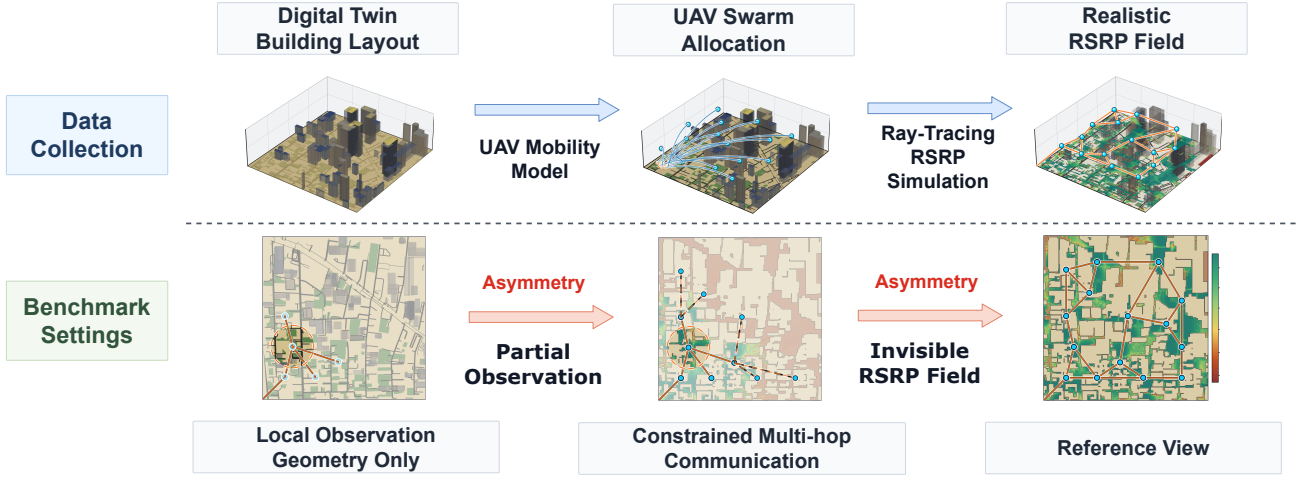}
\vspace{-2pt}
\caption{RadioDynamics setting for deployable UAV-swarm coverage.
The upper flow constructs the radio environment from digital-twin geometry,
UAV mobility,
and ray-traced Reference Signal Received Power (RSRP) fields.
The lower flow contrasts deployment-visible local geometry and constrained
multi-hop communication with a full-information reference view.
The marked \textcolor{red}{asymmetries} highlight the gaps between visible geometry and hidden radio impact,
and between sparse local communication and swarm-level coverage objectives.}
\label{fig:real}
\vspace{-8pt}
\end{figure*}

\subsection{Paper Organization}
The remainder of this paper is organized as follows.
Section~\ref{sec:system_model} presents the system model and formulates the problem.
Section~\ref{sec:empirical_analysis} diagnoses the counterfactual composition gap in synchronized counterfactual deployment and motivates multi-round negotiation. 
Section~\ref{sec:wonder_method} introduces the WONDER methodology, while 
Section~\ref{sec:theoretical_analysis} proves that sequential commitment admits a recursive advantage over the parallel synchronized counterfactual plan. 
Section~\ref{sec:experiments} reports the experimental evaluation,
and Section~\ref{sec:conclusion} concludes the paper.
\section{System Model and Problem Formulation}
\label{sec:system_model}
Beforehand, the main notations used are summarized in Table \ref{tab:notation}.
Particularly, calligraphic letters denote neural networks, objectives, and loss functions; 
uppercase letters denote constants; 
lowercase letters denote variables; and Greek letters denote parameters.

\MainNotationTable
\vspace{-15pt}
\subsection{System Model}
\label{subsec:system_model}
We consider a finite-horizon UAV-swarm coverage problem for post-disaster service restoration, where UAV positions induce the Air-to-Ground (A2G) service field and Air-to-Air (A2A) communication graph.
Let $\UAVSet=\{1,\ldots,\NumUAV\}$ denote the UAV set, and let $\SeviceMap$ be the sampled service region specified by the digital-twin layout $\EnvMap$ (e.g., boundary between buildings and service ground, and building height maps).
At deployment step $t$, UAV $i\in\UAVSet$ has position $\UAVPos_{i,t}\in\mathbb R^3$.
The joint UAV configuration induces the radio field $\RadioField_t$ and $L_t$ A2A disjoint subgraph sets $\GraphSet = \Sigma_{l=1}^{L_t}\{\SubgraphSet, \EdgeSet\}$, 
yielding the global state
\begin{equation}
\SwarmState_t
=
\left(
\{\UAVPos_{i,t}\}_{i\in\UAVSet},
\GraphSet,
\RadioField_t
\right).
\label{eq:physical_state_decomposition}
\end{equation}
Subsequently, each UAV $i$ acts on local observations $\LocalObs_i$ and coordination messages $\{\Msg\}_{j\in \GraphSet}$ exchanged over the A2A graph. 
The following models specify the radio service metrics, communication topology, and decentralized decision boundary.

\subsubsection{Communication Model}
\label{subsec:communication_coverage}
At step $t$, for each ground position $x\in\SeviceMap$, 
geometry-based ray tracing on $\EnvMap$ gives the aggregate A2G power gain $\Gain^{\mathrm{A2G}}(x)$ from UAV $i$ to $x$~\cite{sionnaRT2025}.
The received reference power is
$
\GroundPower_{i,t}(x)
=
\RefSignalPower{i}
\Gain_{i,t}^{\mathrm{A2G}}(x)$,
where $\RefSignalPower{i}$ denotes the reference-signal transmit power of UAV $i$.
The serving UAV is selected by $i_t^\star(x)=\arg\max_{i\in\UAVSet}\GroundPower_{i,t}(x)$, 
and the resulting field $\GroundPower_{i,t}$ is used to define the ground-service metrics below. The A2G ground-service field is summarized by 
signal strength $\RSRP_t(x)$, link quality $\SINR_t(x)$, and rate $\Rate_t(x)$:
\begin{equation}
\begin{aligned}
\RSRP_t(x)
&=
10\log_{10}\!\left(
\frac{\GroundPower_{i_t^\star(x),t}(x)}{1\,\mathrm{mW}}
\right),\\
\SINR_t(x)
&=
\frac{\GroundPower_{i_t^\star(x),t}(x)}
{\sigma_x^2+\Interference_t^{\mathrm{A2G}}(x)},\\
\Rate_t(x)
&=
\Bandwidth^{\mathrm{A2G}}
\log_2\!\left(1+\SINR_t(x)\right).
\end{aligned}
\label{eq:ground_service_metrics}
\end{equation}
Here, $\RSRP_t(x)$ follows the RSRP convention in~\cite{3gpp38215}, $\sigma_x^2$ is the receiver noise power, 
$\Interference_t^{\mathrm{A2G}}(x)$ is the aggregate co-channel interference from non-serving UAVs on the reused A2G downlink resource, and $\Bandwidth^{\mathrm{A2G}}$ is the allocated A2G bandwidth.

The A2A graph $(\SubgraphSet,\EdgeSet)$ specifies one-hop message exchange between UAV $i\in \SubgraphSet$ and its neighbors 
$\UAVNeighbor_i=\{j:(i,j)\in\EdgeSet, i,j \in \SubgraphSet \}$.
For distinct UAVs $i$ and $j$, the layout $\EnvMap$ determines the distance 
$d_{ij,t}=\|\UAVPos_{i,t}-\UAVPos_{j,t}\|_2$.
Consistent with the Line of Sight (LoS) and Non-LoS (NLoS) path-loss functions in~\cite{etsi3gpp38901}, the A2A path loss is
\begin{equation}
\begin{aligned}
\mathrm{PL}_{ij,t}^{\mathrm{A2A}}
&=\ell_{ij,t}\mathrm{PL}^{\mathrm L}\!\left(d_{ij,t},f_c\right)
+\left(1-\ell_{ij,t}\right)\mathrm{PL}^{\mathrm N}\!\left(d_{ij,t},f_c;h_{\mathrm{eff}}\right),
\end{aligned}
\label{eq:a2a_geometry_pathloss}
\end{equation}
where the LoS indicator $\ell_{ij,t}
=\operatorname{LOS}(\UAVPos_{i,t},\UAVPos_{j,t};\EnvMap)$, 
$f_c$ is the carrier frequency, and $h_{\mathrm{eff}}$ is the effective UAV antenna height.
With antenna gains $G_i$ and $G_j$, the directional received power is\footnote{Different from the A2G path, in the A2A and backhaul link budgets, transmit and received powers and link thresholds are expressed in dBm,
while antenna gains and path losses are expressed in dB.}
\[
P_{i\to j,t}^{\mathrm{rx}} = P_i^{\mathrm{A2A}} +G_i+G_j -\mathrm{PL}_{ij,t}^{\mathrm{A2A}}.
\]
Because coordination requires bidirectional exchange, $(i,j)$ is included in $\EdgeSet$ only if both $P_{i\to j,t}^{\mathrm{rx}}$ and $P_{j\to i,t}
^{\mathrm{rx}}$ exceed a threshold $\LinkThresh$.

Backhaul access is provided by a fixed boundary gateway $e$. 
We collect the UAVs with direct backhaul access into the root set $\RootSet_t$. 
Under the backhaul path loss $\mathrm{PL}_{i,t}^{\mathrm{bh}}$ in~\cite{xu2026scalable},
the direct-access root set is
\begin{equation}
\begin{aligned}
\RootSet_{t}
=\{i\in\UAVSet:\;&
P_e^{\mathrm{bh}}
+G_e+G_i -\mathrm{PL}_{i,t}^{\mathrm{bh}}
\geq\BackhaulThresh\}.
\end{aligned}
\label{eq:backhaul_root_set}
\end{equation}
The UAVs outside $\RootSet_{t}$ obtain core-network access through multi-hop paths in $\SubgraphSet$ to this root set.
A UAV is gateway-connected if it has a multi-hop A2A path to at least one root in $\RootSet_{t}$, denoted as $\ConnIndicator_i(\SwarmState_t;e) = 1$; otherwise, $\ConnIndicator_i(\SwarmState_t;e) = 0$.
The gateway-connectivity level is
\begin{equation}
\Connectivity(\SwarmState_t;e)
=
\frac{1}{|\UAVSet|}
\sum_{i\in\UAVSet}\ConnIndicator_i(\SwarmState_t;e).
\label{eq:connectivity_metric}
\end{equation}
At synchronization states, the hard gateway-connectivity requirement is $\Connectivity(\SwarmState_t;e)=1$.

\subsubsection{Decision Model}

At each step $t$, UAV $i$ constructs a deployment-visible observation $\LocalObs_{i,t}\in\mathbb R^{160}$ via one-hop aggragetion, 
including: 
(i) the shared digital-twin geometry $\EnvMap$; 
(ii) the UAV position $\UAVPos_{i,t}$; 
(iii) the one-hop communication topology induced by $\EdgeSet$;
(iv) the relative positions of one-hop neighbors $\{\UAVPos_{j,t}: j \in
\UAVNeighbor_i\}$; 
(v) the geometry around the one-hop neigbors $\{\EnvMap^j: j\in \UAVNeighbor_i \}$\footnote{Each group is encoded into a $32$-dimensional representation and concatenated to form the $160$-dimensional input, i.e., $5 \times 32 = 160$.}.
The global radio field $\RadioField_t \in \mathbb R^{\SeviceMap}$ and non-neighbor states are unavailable to $\pi$. 
Every $\SyncPeriod$ steps, a synchronization step $t$ invokes $\RoundLimit$ rounds of communications, 
%
and UAV $i$ then receives proposal messages $\Msg_{j\to i,t}$ from neighbors $j \in \UAVNeighbor_i$. 
Each message from $j$ to $i$ encompasses the top-$\NumCand$ proposals of trajectories $\Traj_{j, k}$:
\begin{equation}
\Msg_{j\to i,t}
=
\left\{
\left(
j,
\Traj_{j,k},
z_{j,k},
\ProposalScore_{j,k}
\right)
\right\}_{k=1}^{\NumCand},
\label{eq:proposal_packet}
\end{equation}
where ${z}_{j,k}$ indicates a latent representation of a movement trajectory $\Traj_{j, k}$ in terms of the expected global metric and $\ProposalScore_{j,k}$ corresponds to the related evaluation score. 

UAVs then make joint plan $\JointPlan_t = \{\PrimAct_{i,t}\}^{t:t+\SyncPeriod}_{i \in \UAVSet}$ for next $\SyncPeriod$ steps via deployment-available information $\DeployInfo_{i,t}$: 
\begin{equation}
\DeployInfo_{i,t}
\triangleq
\left(
\LocalObs_i,
\{m_{j\to i,t}\}_{j\in\UAVNeighbor_i}
\right).
\label{eq:decentralized_deployment_information}
\end{equation}
which induces
\begin{equation}
\SwarmState_{t+\SyncPeriod}
\sim
\TransKernel\!
\left(\cdot\mid\SwarmState_t,\JointPlan_t;\EnvMap\right).
\label{eq:state_transition}
\end{equation}

\subsection{Problem Formulation}
\label{subsec:problem_formulation}

Each restoration process consists of a warm-up prefix followed by a decentralized policy $\pi$. 
During the first $\WarmupSteps$ steps, the swarm executes a heuristic flocking policy \cite{atincc2020swarm, wu2021multi, capelli2020connectivity, lin2021online}.
%
Particularly, as shown in Fig. \ref{fig:evolution_start}, a swarm of UAVs gradually disperses 
while maintaining inter-UAV connectivity.
At each step $t$ in this prefix, UAV $i$ selects
\begin{equation}
\begin{aligned}
\WarmupAct_{i,t}
&=
\operatorname*{arg\,max}_{a\in\ActSet_i(\SwarmState_t)}
\Bigl[
\DispatchSepWeight\DispatchSep_{i,t}(a)
-
\DispatchCrowdWeight\DispatchCrowd_{i,t}(a) \\ 
&\hspace{6.3em}
-
\DispatchDegradeWeight\DispatchDegrade_{i,t}(a)
\Bigr],
\end{aligned}
\label{eq:warmup_dispatch}
\end{equation}
where $\ActSet_i(\SwarmState_t)$ is a five-dimensional discrete action space, consisting of four cardinal movements and a hovering action, 
$\DispatchSepWeight$, $\DispatchCrowdWeight$ and $\DispatchDegradeWeight$ denote the weight factors, while $\DispatchSep_{i,t}$, $\DispatchCrowd_{i,t}$ and $\DispatchDegrade$ denote the action-related separation reward, crowding penalty and connectivity-degradation penalty, respectively\cite{atincc2020swarm, wu2021multi, capelli2020connectivity, lin2021online}.
%
Although this heuristic quickly disperses the UAVs, as evidenced in Fig. \ref{fig:evolution} lately, the resulting warm-up trajectory still leaves substantial room for optimization due to the lack of radio-aware gain assessment.  

\begin{figure}[t]
\centering
\includegraphics[width=0.8\linewidth]{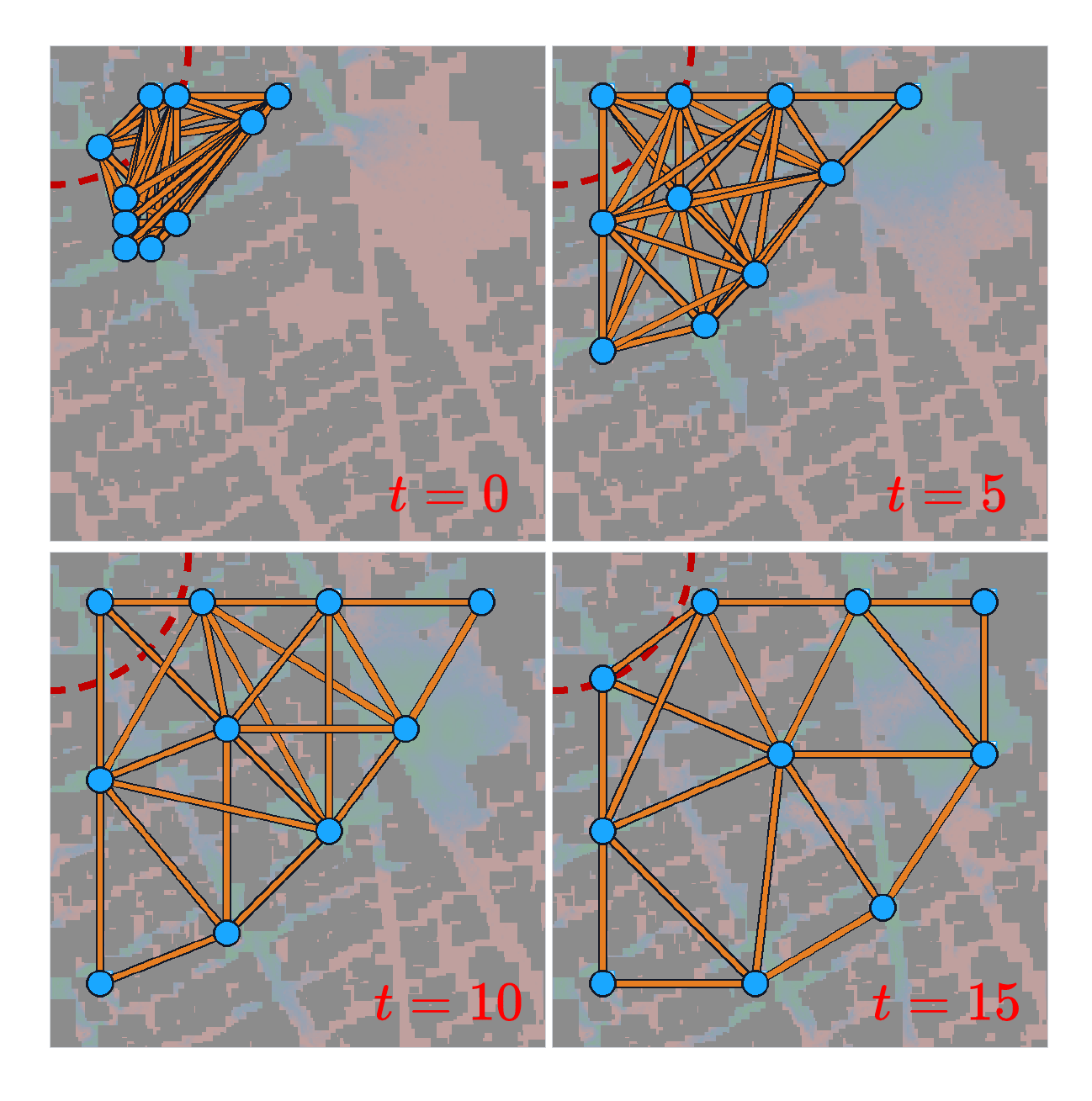}
\vspace{-15pt}
\caption{Visualization of flocking warm-up prefix, with $\WarmupSteps = 15$. 
In each panel, \textcolor{blue}{blue} markers denote UAV positions, \textcolor{orange}{orange} segments denote feasible inter-UAV communication links, 
and the \textcolor{red}{red} circle in the upper left denotes the terminal backhaul range; UAVs outside this range depend on multi-hop UAV relaying for backhaul connectivity.
\textcolor{gray}{dark-gray} regions indicate building footprints, 
and the \textcolor{red}{red}--\textcolor{green}{green} heat map encodes RSRP strength.
}
\vspace{-15pt}
\label{fig:evolution_start}
\end{figure}
Contingent on the heuristic initialization, a decentralized policy $\pi$, which follows the deployment-available information in \eqref{eq:decentralized_deployment_information}, then runs for $\LearnSteps$ synchronization rounds.  
Let ${\SwarmState_\pi}$ denote the reported state obtained after accomplishing the complete restoration process, 
and the metric vector $\MetricVec({\SwarmState_\pi};e)=[\mu_1,\ldots,\mu_7]^{\top}$ denotes the seven evaluation metrics $\pi$ induced by $\pi$. 
The connectivity-preserved coverage optimization $J$ objective can be written as
\begin{align}
\max_{\pi}&\quad
J = \mathbb E_{\pi}\!\left[\MetricWeight^{\top}
\MetricVec\!\left({\SwarmState_\pi};e\right)\right] \nonumber\\
\text{s.t.,  } &
\Connectivity({\SwarmState_\pi};e)\equiv 1,
\forall t
\label{eq:constrained_deployment_problem}\\
& \mu_1=\left\langle
\mathbbm 1[{\RSRP_\pi}\geq\CoverageThresh]
\right\rangle_{\SeviceMap}
\quad
\mu_2=\left\langle
10\log_{10}{\SINR_\pi}
\right\rangle_{\SeviceMap},
\nonumber\\
& \mu_3=\left\langle
\mathbbm 1[10\log_{10}{\SINR_\pi}\geq\SINRThresh]
\right\rangle_{\SeviceMap}
\quad
\mu_4=\left\langle
{\Rate_\pi}
\right\rangle_{\SeviceMap},
\nonumber\\
& \mu_5=
\frac{\mu_4^2}{\left\langle{\Rate_\pi}^{\,2}\right\rangle_{\SeviceMap}}
\,\,
\mu_6=\Connectivity({\SwarmState_\pi};e)
\,\, 
\mu_7=\left\langle
{\Interference_\pi}^{\mathrm{A2G}}
\right\rangle_{\SeviceMap}.\nonumber
\end{align} 
where $\MetricWeight$ sets the evaluation weights and 
the operator $\langle f\rangle_{\SeviceMap}\triangleq |\SeviceMap|^{-1}\sum_{x\in\SeviceMap} f(x)$. This paper aims to find an appropriate policy $\pi$ that maximizes the objective in \eqref{eq:constrained_deployment_problem}. 
\section{Empirical Analysis}
\label{sec:empirical_analysis}

We first evaluate the feasibility of STACCA~\cite{sinha2025stacca},
a recent and representative CTDE counterfactual method,
in the RadioDynamics scenario.
The method uses shared context to evaluate each local decision while holding the other agents' decisions fixed.
Its centralized counterfactual critic strengthens this branch-wise evaluation by amplifying the credit-assignment signal available to decentralized actors.
As shown in Fig.~\ref{fig:STACCA_curve},
this baseline exhibits oscillatory dynamics in actor optimization and return on the RadioDynamics UAV coverage task,
suggesting a mismatch between counterfactual branch evaluation and synchronized long-horizon deployment. We adopt the following definition to analyze the reasons.

\begin{figure}[!tb]
\centering
\includegraphics[width=0.98\linewidth]{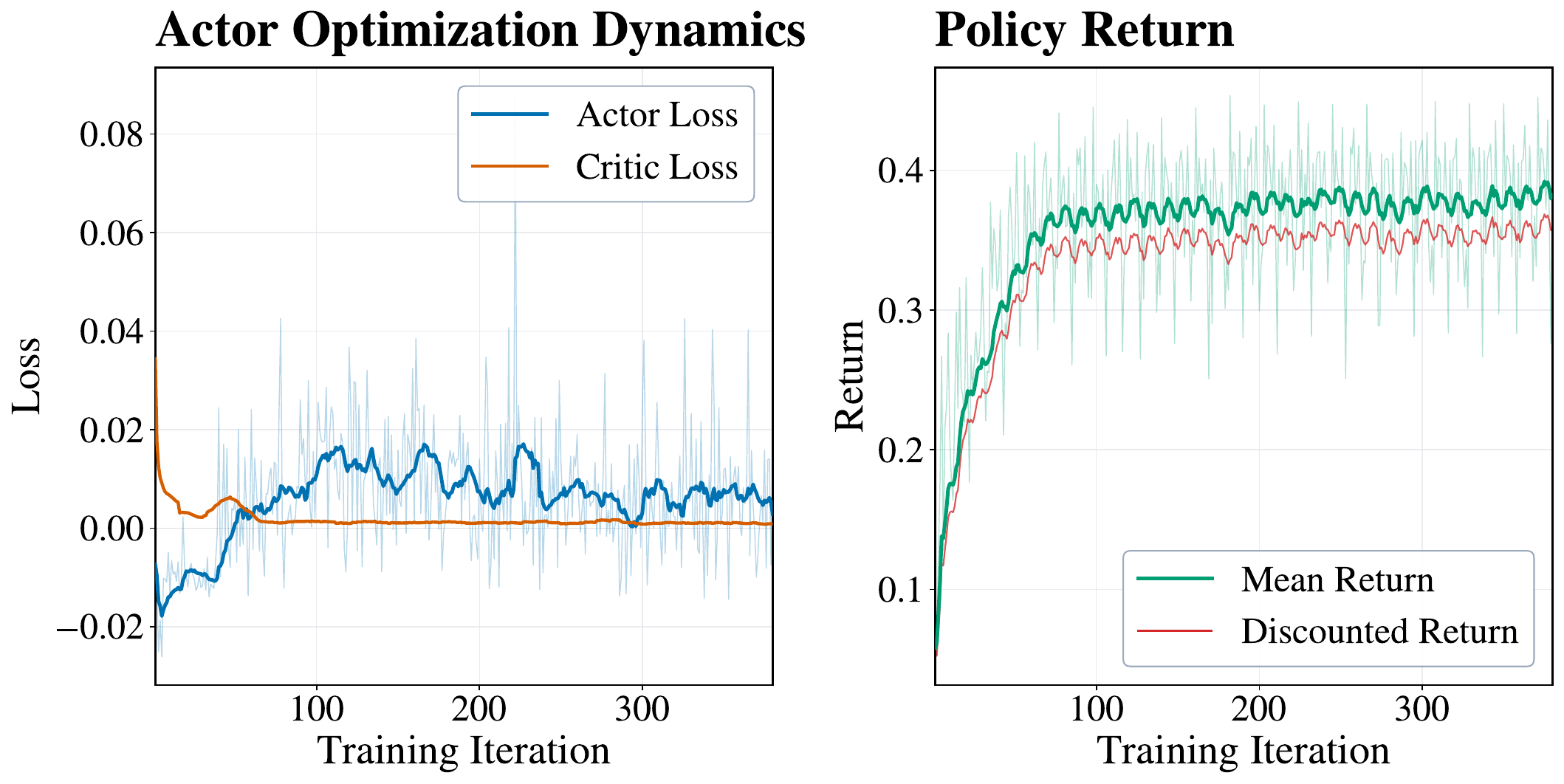}
\vspace{-8pt}
\caption{Loss curve and return curve of STACCA in RadioDynamics.}
\label{fig:STACCA_curve}
\end{figure}

\begin{figure}[!tb]
\centering
\includegraphics[width=0.98\linewidth]{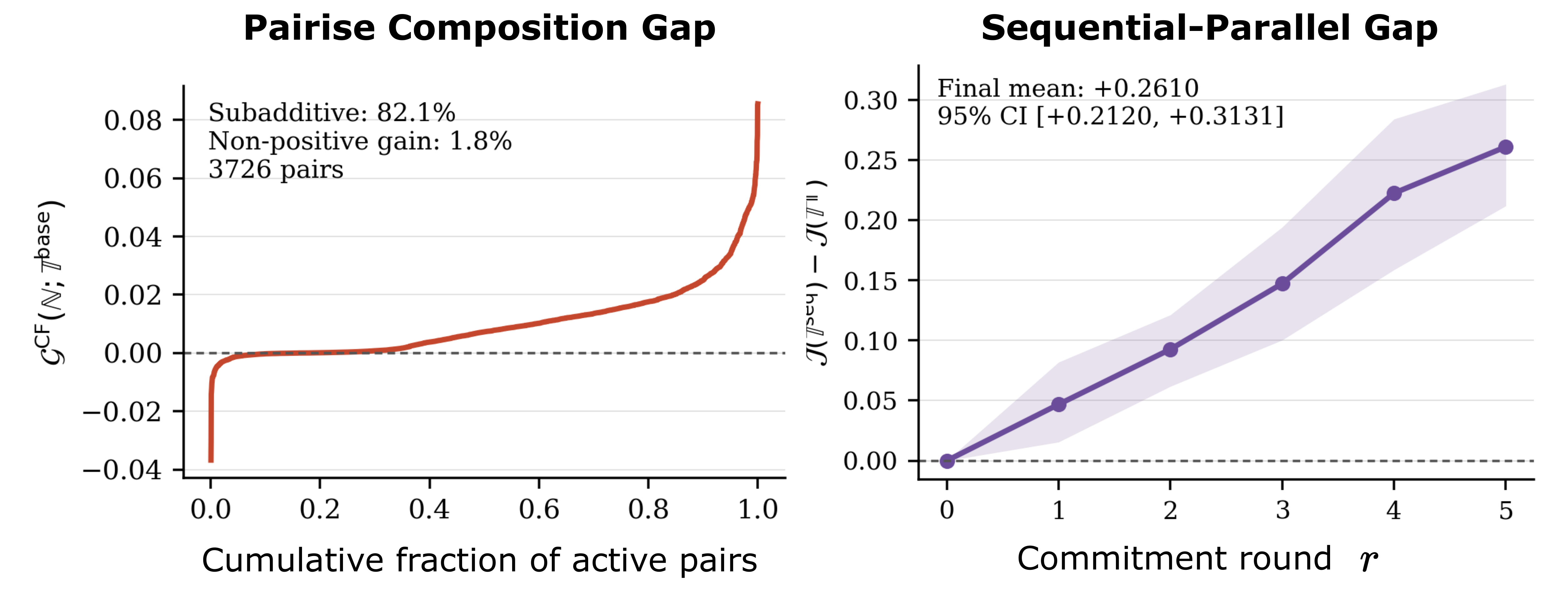}
\vspace{-8pt}
\caption{Empirical evidence for counterfactual composition gap and sequential-plan gains in RadioDynamics. 
Both experiments are implemented across \(51\) train and validation scenes. 
Left: pairwise \(\CompositionGap(\UAVSet;\BasePlan)\) measured from \(3,726\) active UAV pairs, comparing synchronized two-UAV counterfactual estimates with the impact of parallel execution. 
Right: the incremental balanced score of \(\SeqPlan_r\) over \(\ParallelPlan\) increases with the negotiation round \(r\); the shaded region denotes the \(95\%\) confidence interval.}
\vspace{-15pt}
\label{fig:residual_curve}
\end{figure}

\begin{definition}[Counterfactual composition gap]
\label{def:counterfactual_composition_gap}
For each UAV \(i\) with reference plan $\BasePlan$, the counterfactual trajectory selected from its feasible set $\FeasProposalPool_r$ is
\[
\CounterfactualTraj_{i}
\in
\arg\max_{\Traj_{i,k}\in\FeasProposalPool_r}
\CFGain(\Traj_{i,k};\BasePlan).
\]
This selected trajectory is the replacement trajectory used to update the reference plan in the following parallel or
sequential construction.
The composed parallel plan over this set is
\begin{equation}
\ParallelPlan
=
\BasePlan\oplus((i,\CounterfactualTraj_i))_{i\in\UAVSet}.
\label{eq:ParallelGain}
\end{equation}
The counterfactual composition gap is
\begin{equation}
\CompositionGap(\UAVSet;\BasePlan)
=
\sum_{i\in\UAVSet}
\CFGain(\CounterfactualTraj_i;\BasePlan)
-
\left[
\BalancedScore(\ParallelPlan)
-
\BalancedScore(\BasePlan)
\right].
\label{eq:CompositionGap}
\end{equation}
A positive \(\CompositionGap(\UAVSet;\BasePlan)\) means that individually improving counterfactual trajectories
overestimates the gain realized when those trajectories are executed together.
\end{definition}

To narrow the composition gap $\CompositionGap$, 
we instantiate sequential plan $\SeqPlan$ with the same counterfactual evaluation as $\ParallelPlan$, but commits one trajectory per round:
\begin{align}
\CommittedTraj_r
&=
\operatorname{argmax}_{\CounterfactualTraj_i,\,i\in\UAVSet\setminus\CommittedUAVSet_r}
\CFGain(\CounterfactualTraj_i;\BasePlan\oplus\CommittedTrajSet{r}).
\nonumber
\end{align}
After \(r\) rounds, the resulting plan is
\begin{align}
\SeqPlan_{r}
&=
\CommittedTrajSet{r}
\oplus((i,\CounterfactualTraj_i))_{i\in\mathbb N^{\mathrm{CF}}_{r}}.
\quad
\SeqPlan_0=\ParallelPlan.
\nonumber
\end{align}

Fig.~\ref{fig:residual_curve} provides the empirical basis for the subsequent implementation and analysis. 
The left panel shows that synchronized counterfactual overestimates parallel execution: \(82.1\%\) of the \(3,726\) active UAV pairs have \(\CompositionGap(\UAVSet;\BasePlan) > 0\), with a mean gap of \(1.00\%\) and a median gap of \(0.73\%\).
The right panel show that \(\BalancedScore(\SeqPlan_r)-\BalancedScore(\ParallelPlan)>0\) for every \(r\ge1\) and increases from \(0.047\) at \(r=1\) to \(0.261\) at \(r=5\), with all \(95\%\) confidence intervals remaining positive.
These observations motivate WONDER's multi-round negotiation, which commits one trajectory per round and recomputes the remaining proposals under the updated context. 
Sec.~\ref{sec:theoretical_analysis} formalizes this design by showing that \(\SeqPlan_r\) admits a larger advantage than \(\ParallelPlan\) as \(r\) increases. 

\begin{figure*}[!htb]
\centering
\includegraphics[width=0.98\textwidth]{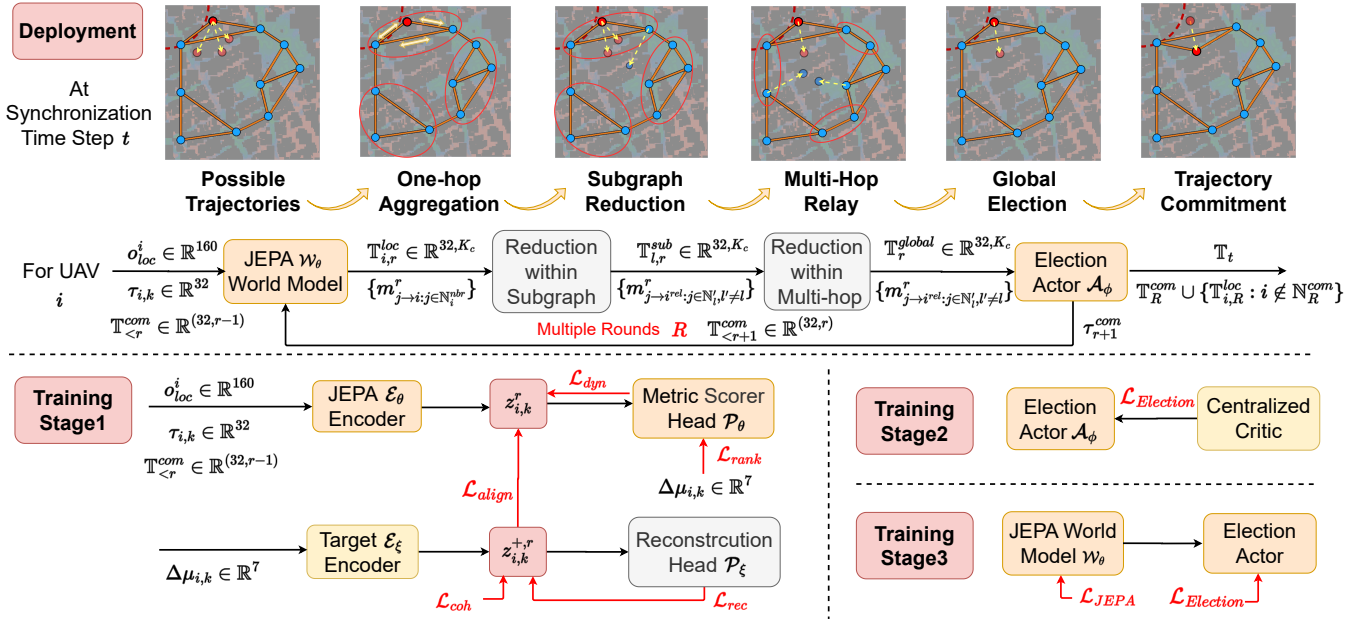}
\vspace{-12pt}
\caption{Architecture of WONDER. The upper panel shows the deployment-time inference pipeline at a synchronization time step, 
including candidate trajectory generation, one-hop aggregation, subgraph reduction, multi-hop relay, global election, and trajectory commitment. 
The lower panel shows the three-stage training pipeline for the JEPA world model and election actor. 
\textcolor{orange}{Orange} blocks denote modules used at deployment, and \textcolor{yellow}{yellow} blocks denote modules used only during training.}
\label{fig:module}
\end{figure*}
\section{Methodology}
\label{sec:wonder_method}

\subsection{Overview}
\label{subsec:wonder_overview}
At each synchronization step $t$, 
WONDER constructs a joint plan $\JointPlan_t$ through multiple commitment rounds of trajectory election, 
where each round commits one elected trajectory and action selection is performed only after the plan is finalized.
Overall, as shown in Fig.~\ref{fig:module}, the candiate trajectories 
undergo a reduction within each subgraph $\SubgraphSet$, 
a multi-hop relay, global election to derive a final commitment.  
Based on available information $\DeployInfo_{i,t}$, each UAV filters endpoints reachable within $\SyncPeriod$ steps under the Manhattan distance and constructs non-overlapping feasible trajectories $\Traj_{i, k} \in \FeasProposalPool_{i,r}$ to the retained endpoints. 
From the perspective of UAV $i$, as specified in \eqref{eq:proposal_packet}, the message $\Msg_{j\to i}^{r}$ received from UAV $j$ contains a set of trajectories.
At round $r$, the already committed UAVs form the set $\CommittedUAVSet_r$, whose trajectories form $\CommittedTrajSet{r}$, and the remaining UAVs form the set $\ActiveUAVSet_r$. 
The JEPA radio world model $\WorldModel$, the details of which will be discussed in Section \ref{subsec:radio_aware_world_model}, ranks these trajectories and retains the top $\NumCand$ candidate trajectories $\LocalProposalPool_{i,r}$.
For UAV $i$, neighbor messages $\{\Msg_{j\to i}^{r}:j\in\UAVNeighbor_i\}$ collect the trajectories available inside the same $\SubgraphSet$.
Each $\SubgraphSet$ applies a deterministic $\mathsf{Top}_{\NumCand}$ reduction to these messages and 
forms the subgraph trajectory pool $\SubgraphProposalPool_r$.
For a relay UAV $i^{\mathrm{rel}}\in\UAVSet_l$, 
inter-subgraph relay makes the messages $\{\Msg_{j\to i^{\mathrm{rel}}}^{r}:j\in\UAVSet_{l'},\,l'\neq l\}$ available within at most $\RelayDepth$ relay hops.
This relay step forms the global trajectory pool $\GlobalProposalPool_r$ with $\NumCand$ trajectories for the PPO-based election actor $\ElectionModule$.
The elected trajectory determines the committed trajectory $\CommittedTraj_r$, which is appended to $\CommittedTrajSet{r+1}$. 
After $R$ negotiation rounds, the joint plan is assembled as $ \JointPlan_t=\CommittedTrajSet{R}\cup\{\LocalProposalPool_{i,r} : i\notin\CommittedUAVSet_R\}$, combining the committed trajectories with the local proposals of UAVs that remain uncommitted.

As illustrated in Fig.~\ref{fig:module}, the message-driven trajectory chain in round $r$ is summarized as
%
\begin{align}
&\FeasProposalPool_{i,r}, \CommittedTrajSet{r}
\xrightarrow{\ \WorldModel}
\LocalProposalPool_{i,r}
\xrightarrow{\ \{\Msg_{j\to i}^{r}\}_{j\in\UAVNeighbor_i}}
\SubgraphProposalPool_{l,r} \nonumber\\
&\xrightarrow[
\RelayDepth
]{\{\Msg_{j\to i^{\mathrm{rel}}}^{r}:j\in\UAVSet_{l'},\,l'\neq l\}}
\GlobalProposalPool_{r}
\xrightarrow{\ \ElectionModule\ }
\CommittedTraj_{r} 
\xrightarrow{}
\CommittedTrajSet{r+1}\label{eq:proposal_chain}
\end{align}


\subsection{Radio World Model}
\label{subsec:radio_aware_world_model}
Under the deployment information $\DeployInfo_{i,t}$, the future radio field $\RadioField_t$ and non-neighbor states are unavailable, making direct evaluation of each trajectory infeasible at deployment. 
WONDER therefore uses a radio world model $\WorldModel$ to estimate the decision-relevant radio impact of each feasible trajectory. 

As shown in Fig.~\ref{fig:module}, $\WorldModel$ contains a deployment encoder $\DeployEncoder$, a metric-change prediction head $\DeployPredictor$, and a rollout-supervised target encoder $\TargetEncoder$.
The prediction is conditioned on the evolving negotiation context. 
At round $r$, let $\JointPlan_t^{r}$ denote the partial joint plan formed by the trajectories in $\CommittedTrajSet{r}$. 
Each new commitment $\CommittedTraj_r$ updates this context for the remaining candidates. 
Across negotiation rounds, $\WorldModel$ therefore predicts each trajectory's $\Traj_{i,k}^{r}$ incremental impact 
relative to the trajectories 
already committed.
%
%
The scalar trajectory score $\ProposalScore_{i,k}$ used by local, subgraph, and relay reductions is then computed from the predicted metric change.
Here, a metric $\MetricVec$ can be a composition of coverage area, Jain fairness, and connectivity mentioned in Eq.~\eqref{eq:constrained_deployment_problem}. 
%
%
Formally, the rollout target is computed as: 
\begin{equation}
\begin{aligned}
\Delta\MetricVec_{i,k}^{r}
=
\MetricVec\!\left(
\SwarmState_{t+\SyncPeriod}; e
\right)\Big|_{\JointPlan_t^{r}\oplus(i,\Traj_{i,k})}
-
\MetricVec\!\left(
\SwarmState_{t+\SyncPeriod};e
\right)\Big|_{\JointPlan_t^{r}} .
\end{aligned}
\label{eq:target_metric_change}
\end{equation} 

The target encoder $\TargetEncoder$ maps the rollout-measured metric change $
\Delta\MetricVec_{i,k}^{r}$ to the target latent $z_{i,k}^{+,r}$, and the target
predictor $\TargetPredictor$ reconstructs this change as $
\widehat{\Delta\MetricVec}_{i,k}^{+,r}$:
\begin{equation}
z_{i,k}^{+,r}
=
\TargetEncoder\!\left(
\Delta\MetricVec_{i,k}^{r}
\right), \ \widehat{\Delta\MetricVec}_{i,k}^{+, r}
=
\TargetPredictor\!\left(
z_{i,k}^{+,r}
\right). 
\label{eq:jepa_world_model}
\end{equation}
The deployment encoder $\DeployEncoder$ maps the deployment-visible information $
\DeployInfo_{i,t}$ to the deployment latent $z_{i,k}^{r}$, from which the
deployment predictor estimates the metric change $\widehat{\Delta\MetricVec}_{i,k}
^{r}$:
\begin{equation}
z_{i,k}^{r}
=
\DeployEncoder\!\left(
\DeployInfo_{i,t},\Traj_{i,k},\CommittedTrajSet{r}
\right), \ \widehat{\Delta\MetricVec}_{i,k}^{r}
=
\DeployPredictor\!\left(
z_{i,k}^{r}
\right). 
\end{equation}
First, both paths are supervised by the rollout-measured metric change $\Delta\MetricVec_{i,k}^{r}$.
The target path reconstructs this metric change from the target latent $z_{i,k}^{+,r}$, while the deployment path predicts the same metric change from the deployment latent $z_{i,k}^{r}$:
\begin{align}
\TargetReconLoss
& =\mathbb E_{i,k,r}
\!\left[
\HuberLoss\!\left(
\widehat{\Delta\MetricVec}_{i,k}^{r},
\Delta\MetricVec_{i,k}^{r}
\right)
\right]\nonumber\\
&=
\mathbb E_{i,k,r}
\!\left[
\HuberLoss\!\left(
\widehat{\Delta\MetricVec}_{i,k}^{+, r},
\Delta\MetricVec_{i,k}^{r}
\right)
\right].\label{eq:target_reconstruction_loss}
\end{align}
where $\HuberLoss$ denotes a Huber loss.
Second, the alignment loss transfers the rollout-defined representation to the deployment encoder by matching the deployment latent to the stopped target latent:
\begin{equation}
\AlignLoss
=
\mathbb E_{i,k,r}
\!\left[
\left\|
\LatentNorm{z_{i,k}^{r}}
-
\operatorname{sg}\!\left(\LatentNorm{z_{i,k}^{+,r}}\right)
\right\|_2^2
\right],
\label{eq:alignment_loss}
\end{equation}
where $\operatorname{sg}(\cdot)$ denotes the stop-gradient operator and $\LatentNorm{\cdot}$ denotes the  latent-vector $\ell_2$ normalization operator.
Third, the trajectories whose endpoints are nearby are encouraged to have consistent target latents. 
Let $\ProposalEndpoint_{i,k}^{r}$ denote the endpoint position of trajectory $\Traj_{i,k}$, 
the coherent loss is 
\begin{equation}
\CoherenceLoss
=
\mathbb E_{i,k,k',r}
\!\left[
\exp\!\left(-\frac{\EndpointDist{k}{k'}}{\CoherenceTemp}\right)
\LatentDist{k}{k'}^2
\right],
\label{eq:coherence_loss}
\end{equation}
where
$
\EndpointDist{k}{k'}
=
\|\ProposalEndpoint_{i,k}^{r}-\ProposalEndpoint_{i,k'}^{r}\|_1$ and $
\LatentDist{k}{k'}
=
\|\LatentNorm{z_{i,k}^{+,r}}-\LatentNorm{z_{i,k'}^{+,r}}\|_2$.    

Finally, the JEPA training objective can be described: 
\begin{equation}
\JEPALoss
=
\TargetLoss
+
\CoherenceLoss
+
\AlignLoss.
\label{eq:jepa_world_model_objective}
\end{equation}
Further, the scalar trajectory score can be induced with $\ProposalScore_{i,k}^{r} = \MetricWeight^{\top} \widehat{\Delta\MetricVec}_{i,k}^{r}$, and an auxiliary ranking supervision $\RankLoss$ can be used for its training; details are given in Eq.~\eqref{eq:ranking_loss} of Appendix \ref{app:jepa_loss}.
  
\begin{algorithm}[t]
\caption{Inference in WONDER.}
\label{alg:inference_proposal_election}
\begin{algorithmic}[1]
\Require Local observations $\{\LocalObs_i\}_{i\in\UAVSet}$, communication graph $\GraphSet$, trained $\WorldModel$ and $
\ElectionModule$
\Ensure Joint plan $\JointPlan_t$
\State Initialize $\CommittedTrajSet{1}\gets\emptyset$, $\CommittedUAVSet_1\gets\emptyset$, and $\ActiveUAVSet_1\gets\UAVSet$
\For{$r=1,\ldots,\RoundLimit$}
    \If{$\ActiveUAVSet_r=\emptyset$}
        \State \textbf{break}
    \EndIf
    \For{each UAV $i\in\ActiveUAVSet_r$}
        \State Filter endpoints reachable within $\SyncPeriod$ steps and obtain non-overlapping feasible trajectories $\FeasProposalPool_{i,r}$
        \State Rank $\FeasProposalPool_{i,r}$ with $\WorldModel$ \eqref{eq:jepa_world_model} and retain the top $\NumCand$ local proposals
        \State Send trajectory message $\Msg_{i\to j}^{r}$ to each neighbor $j\in\UAVNeighbor_i$
    \EndFor
    \For{each subgraph $\SubgraphSet\in\GraphSet$}
        \State Collect neighbor messages $\{\Msg_{j\to i}^{r}\}_{j\in\UAVNeighbor_i}$
        \State Apply deterministic $\mathsf{Top}_{\NumCand}$ reduction to form $\SubgraphProposalPool_r$
    \EndFor
    \State Relay reduced messages across subgraphs for at most $\RelayDepth$ hops
    \State Form the global trajectory pool $\GlobalProposalPool_r$ following \eqref{eq:proposal_chain}
    \State $\ElectionChoice_r\gets\ElectionModule(\GlobalProposalPool_r,\CommittedTrajSet{r})$
    \If{$\ElectionChoice_r=\text{End Negotiation}$}
        \State \textbf{break}
    \EndIf
    \State Extract the committed trajectory $\CommittedTraj_r$ and selected UAV $\SelectedUAV$ from $\ElectionChoice_r$ and update $\CommittedTrajSet{r+1}\gets\CommittedTrajSet{r}\cup\{\CommittedTraj_r\}$, $\CommittedUAVSet_{r+1}\gets\CommittedUAVSet_r\cup\{\SelectedUAV\}$, $\ActiveUAVSet_{r+1}\gets\ActiveUAVSet_r\setminus\{\SelectedUAV\}$
\EndFor
\State Construct $\JointPlan_t$ from the committed trajectories
\State \Return $\JointPlan_t$
\end{algorithmic}
\end{algorithm}
\subsection{Training Process}
\label{subsec:training_process}

The training process decouples radio impact representation learning from election learning.
WONDER first pretrains the JEPA radio world model $\WorldModel$ and the trajectory scorer $\DeployPredictor$ with offline rollouts.
It then freezes $\WorldModel$ and trains the election module $\ElectionModule$ on $\GlobalProposalPool_r$ constructed from the
current context $\CommittedTrajSet{r}$. 
The election module $\ElectionModule$ selects one trajectory or ends the negotiation, with rewards defined by balanced-objective improvement and gateway-connectivity penalty. We leave the training details in Appendix~\ref{app:election_loss}.
Finally, WONDER couples the evolving radio world model with the Election Actor through alternating proposal refresh. 
Since $\WorldModel$ re-encodes and scores candidates when constructing $\GlobalProposalPool_r$, updating $\WorldParams$ changes the proposal pool seen by $\ElectionModule$. 
Therefore, PPO rollouts collected under the old proposal pool are discarded, 
and the next election update uses fresh rollouts $\FreshRollout$ generated by the updated world model $\WorldModel'$. 

In summary, we provide the inference and training pseudocodes in Algorithm~\ref{alg:inference_proposal_election} and Algorithm~\ref{alg:wonder_training}.

\begin{algorithm}[t]
\caption{Training process of WONDER.}
\label{alg:wonder_training}
\begin{algorithmic}[1]
\Require Offline rollouts, initialized $\WorldModel$, and $\ElectionModule$
\Ensure Trained $\WorldModel$ and $\ElectionModule$
\State \textbf{Stage 1: JEPA radio world model pretraining}
\State Use offline rollouts to compute target metric changes $\Delta\MetricVec_{i,k}^{r}$ in \eqref{eq:target_metric_change}
\State Train the target and deployment encoder and trajectory scorer to $\ProposalScore_{i,k}^{r}$ in \eqref{eq:jepa_world_model} for local, subgraph, and relay reductions with $\JEPALoss$ in \eqref{eq:jepa_world_model_objective}
\State \textbf{Stage 2: Election module training}
\State Freeze $\WorldModel$ 
\While{the election module has not converged}
    \State Run multi-round trajectory election using Algorithm~\ref{alg:inference_proposal_election}
    \State Construct $\GlobalProposalPool_r$ from the current context $\CommittedTrajSet{r}$
    \State Update $\ElectionModule$ with $\ElectionUpdateLoss$ in Eq.~\eqref{eq:app_election_objective}
\EndWhile
\State \textbf{Stage 3: Coupled refinement of world modeling and election}
\While{the refinement has not converged}
  \State Freeze $\ElectionModule$ and update $\WorldModel$ with $\JEPALoss$ in
  \eqref{eq:jepa_world_model_objective}
  \State Regenerate rollouts $\FreshRollout$ with updated $
  \WorldModel'$
  \State Freeze $\WorldModel$ and update $\ElectionModule$ with $
  \ElectionUpdateLoss$ in \eqref{eq:app_election_objective} on $\FreshRollout$
\EndWhile
\State \Return $\WorldModel,\ElectionModule$
\end{algorithmic}
\end{algorithm}

\section{Theoretical Analysis}
\label{sec:theoretical_analysis}

The following theorem connects Definition~\ref{def:counterfactual_composition_gap} with the multi-round negotiation in WONDER.
By comparing a sequential plan \(\SeqPlan\) with the STACCA-style parallel counterfactual plan \(\ParallelPlan\), mentioned in Section~\ref{sec:empirical_analysis}, we show that
sequential commitment yields a larger lower bound on \(\BalancedScore\) as the number of negotiation rounds $r$ increases.
The approximation result for WONDER then explains how the learned policy approaches this bound in the practical negotiation range.

We write the composition gap after round \(r\) as
\begin{align}
\CompositionGap_r
&\triangleq
\CompositionGap(\mathbb N^{\mathrm{CF}}_r;\BasePlan\oplus\CommittedTrajSet{r}).
\nonumber
\end{align}
Here \(\CommittedTrajSet{r}\) contains the trajectories committed before round $r$, and \(\mathbb N^{\mathrm{CF}}_r\) contains the UAVs that recompute counterfactual trajectories under this committed context.

\begin{assumption}[Positive round-wise composition gap]
\label{ass:pos_gap}
For the negotiation rounds considered in the analysis, $\CompositionGap_r \ge 0$.  
\end{assumption}

\noindent The curve of $\CompositionGap(\UAVSet;\BasePlan)$ in Fig.~\ref{fig:residual_curve} provides empirical evidence for this condition.

\begin{assumption}[Composition-gap recovery of the sequential plan]
\label{ass:round_gain}
Committing \(\CommittedTraj_r\) and recomputing the remaining counterfactual trajectories recover at least an
\(\eta_r\)-fraction of the current composition gap:
\begin{align}
\BalancedScore(\SeqPlan_{r+1})-\BalancedScore(\SeqPlan_r)
&\ge
\eta_r\CompositionGap_r,
\qquad 0<\eta_r\le1.
\nonumber
\end{align}
The curve of $\BalancedScore(\SeqPlan_{r}) - \BalancedScore(\ParallelPlan)$ in Fig.~\ref{fig:residual_curve} provides empirical evidence for this condition.
\end{assumption}

Under exact scoring among feasible trajectories $\FeasProposalPool_r$, 
Assumption ~\ref{ass:round_gain} provides a guarantee on the recursive gain of sequential commitment, as formalized below. 
\begin{theorem}[Recursive advantage over parallel counterfactual plan]
\label{thm:ideal}
Let
$D_r
\triangleq
\BalancedScore(\SeqPlan_r)-\BalancedScore(\ParallelPlan)$.
Under Assumption \ref{ass:round_gain},
$D_{r+1}
\ge
D_r+\eta_r\CompositionGap_r.$
Consequently, under Assumption \ref{ass:pos_gap}, if \(\CompositionGap_r\ge0\), the sequential gain over \(\ParallelPlan\) is nondecreasing from round \(r\) to
round \(r+1\). If \(\CompositionGap_r>0\), it increases strictly.
\end{theorem}

\begin{proof}
By Assumption \ref{ass:round_gain},
\begin{align}
\BalancedScore(\SeqPlan_{r+1})-\BalancedScore(\SeqPlan_r)
&\ge
\eta_r\CompositionGap_r.
\nonumber
\end{align}
Adding \(\BalancedScore(\SeqPlan_r)-\BalancedScore(\ParallelPlan)\) to both sides yields
\begin{align}
\BalancedScore(\SeqPlan_{r+1})-\BalancedScore(\ParallelPlan)
&\ge
\BalancedScore(\SeqPlan_r)-\BalancedScore(\ParallelPlan)
+
\eta_r\CompositionGap_r.
\nonumber
\end{align}
Thus, \(D_{r+1}\ge D_r+\eta_r\CompositionGap_r\).  Since \(0<\eta_r\le1\),
\(\CompositionGap_r\ge0\) in Assumption \ref{ass:pos_gap} implies \(D_{r+1}\ge D_r\), and \(\CompositionGap_r>0\) implies \(D_{r+1}>D_r\).
\end{proof}

This guarantee provides a reference for the learned radio world model. 
We next assume sequential proposal effectiveness and bounded election regret, under which WONDER approaches the ideal recursive gain up to a regret \(\epsilon_r\). 

\begin{assumption}[Proposal effectiveness of the sequential choice]
\label{ass:proposal}
For commit rounds \(r\), WONDER's proposal pool contains the trajectory selected by the ideal sequential rule:
\begin{align}
\CommittedTraj_r
\in
\GlobalProposalPool_{r,\WorldModel}.
\nonumber
\end{align}
When WONDER selects from this pool under the same committed context, its chosen trajectory \(\CommittedTraj_{r,\WorldModel}\) incurs at most \(\epsilon_r\) selection regret:
\begin{align}
\CFGain(\CommittedTraj_r;\BasePlan\oplus\CommittedTrajSet{r})
&-
\CFGain(\CommittedTraj_{r,\WorldModel};\BasePlan\oplus\CommittedTrajSet{r})
\nonumber\\
&\le
\epsilon_r,
\qquad \epsilon_r\ge0.
\nonumber
\end{align}
\end{assumption}
\noindent The recall diagnostic in Fig.~\ref{fig:Stage1_Ensemble_Recall} provides empirical evidence for this condition.

We now define the sequential plan $\SeqPlan_{r,\WorldModel}$ induced by WONDER under the learned radio world model.
After $r$ committed rounds, its plan state and one-step update are
\begin{align}
\SeqPlan_{r,\WorldModel}
&=
\BasePlan\oplus\CommittedTrajSet{r},
\nonumber\\
\SeqPlan_{r+1,\WorldModel}
&=
\SeqPlan_{r,\WorldModel}\oplus\CommittedTraj_{r,\WorldModel}.
\nonumber
\end{align}
and define
$D_{r,\WorldModel}
\triangleq
\BalancedScore(\SeqPlan_{r,\WorldModel})-\BalancedScore(\ParallelPlan)$. We can have the following theorem
\begin{theorem}[Recursive gain of WONDER]

Under Assumptions \ref{ass:pos_gap}, \ref{ass:round_gain}, \ref{ass:proposal},
\begin{align}
D_{r+1,\WorldModel}
&\ge
D_{r,\WorldModel}
+
\eta_r\CompositionGap_r
-
\epsilon_r.
\nonumber
\end{align}
\end{theorem}

\begin{proof}
Theorem~\ref {thm:ideal} states that the ideal sequential choice recovers at least \(\eta_r\CompositionGap_r\).  
By Assumption~\ref{ass:proposal}, this trajectory is available in \(\GlobalProposalPool_{r,\WorldModel}\).  
Because \(\SeqPlan_{r,\WorldModel}=\BasePlan\oplus\CommittedTrajSet{r}\), 
Definition~\ref{def:counterfactual_composition_gap} makes \(\CFGain\) the one-step \(\BalancedScore\) increment in the WONDER comparison.  
Assumption~\ref{ass:proposal} further implies that the WONDER step loses at most \(\epsilon_r\) relative to the ideal step.  
Hence, the WONDER step recovers at least \(\eta_r\CompositionGap_r-\epsilon_r\).  
Adding this recovered amount to \(D_{r,\WorldModel}\) gives the recursion. 
\end{proof}

\begin{remark}
By iterating the round-wise recursion and using \(\SeqPlan_0=\ParallelPlan\), the ideal sequential plan satisfies
\begin{align}
\BalancedScore(\SeqPlan_r)-\BalancedScore(\ParallelPlan)
&\ge
\sum_{t=0}^{r-1}\eta_t\CompositionGap_t.
\nonumber
\end{align}
Thus, when the recursive composition gaps are nonnegative, sequential negotiation gives an advantage over the parallel counterfactual plan. Under assumptions on proposal
effectiveness and bounded selection regret, WONDER satisfies the corresponding regret-adjusted bound,
\begin{align}
\BalancedScore(\SeqPlan_{r,\WorldModel})-\BalancedScore(\ParallelPlan)
&\ge
\sum_{t=0}^{r-1}
\left(\eta_t\CompositionGap_t-\epsilon_t\right).
\nonumber
\end{align}
Therefore, WONDER approaches the sequential-plan advantage up to the cumulative selection regret \(\sum_{t=0}^{r-1}\epsilon_t\) in the negotiation rounds where the radio world model is highly confident and the proposal coverage is empirically verified.
\end{remark}

\section{Experiment Settings and Results}
\label{sec:experiments}

\subsection{Experimental Setup}

\begin{figure*}[t]
\centering
\includegraphics[width=0.98\textwidth]{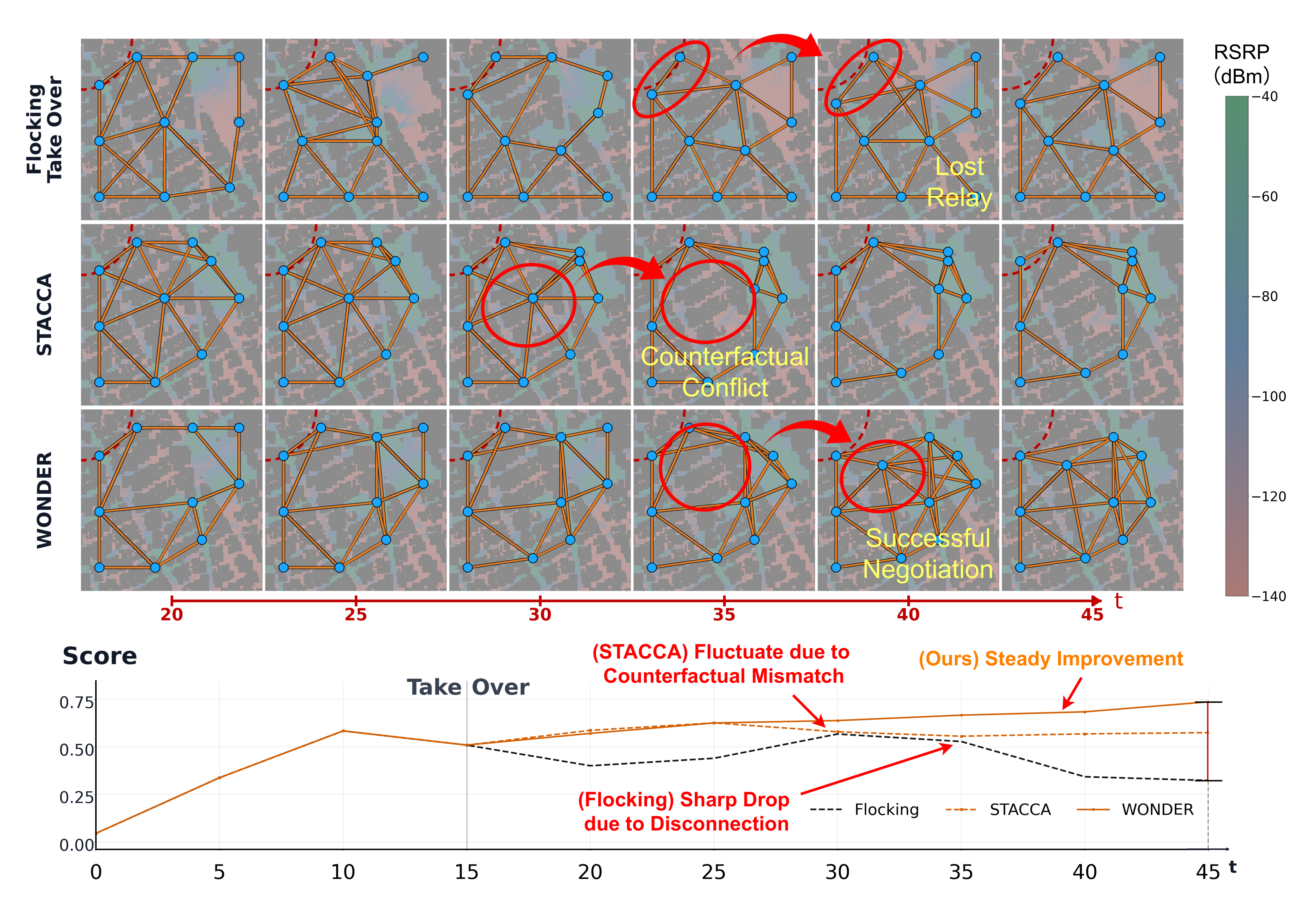}
\vspace{-15pt}
\caption{Visualization of RadioDynamics closed-loop rollouts.
Panels are recorded every five timesteps; the \textcolor{red}{red} time axis aligns the panels with rollout time, and the bottom curve reports the corresponding balanced score.
The first to third rows show rollouts controlled by Flocking, STACCA, and WONDER, respectively.}
\label{fig:evolution}
\vspace{-12pt}
\end{figure*}

\begin{table}[!t]
\centering
\caption{Simulation and benchmark parameters.}
\label{tab:params}
\setlength{\tabcolsep}{2.0pt}
\renewcommand{\arraystretch}{1.05}
\newcommand{\topdesc}[1]{%
  \multicolumn{1}{@{}>{\raggedright\arraybackslash}p{0.38\columnwidth}}{#1}}
\newcommand{\topnote}[1]{%
  \multicolumn{1}{>{\centering\arraybackslash}p{0.25\columnwidth}}{#1}}
\begin{tabularx}{\columnwidth}{@{}>{\raggedright\arraybackslash}m{0.4\columnwidth}>{\centering\arraybackslash}m{0.25\columnwidth}>{\raggedright\arraybackslash}X@{}}
\toprule
Description & Notation & Value \\
\midrule
Ground service area & \(\diagup\) & \(700\times700\) m\(^2\) \\
Ground grid & \(\diagup\) & \(350\times350\) \\
Cell resolution & \(\diagup\) & \(2\) m \\
A2G bandwidth & \(\Bandwidth\) & \(100\) MHz~\cite{mosoCanopy5G} \\
A2A carrier frequency & \(f_c\) & \(3.5\) GHz~\cite{etsi3gpp38901} \\
A2A transmit power & \(P_i^{\mathrm{A2A}}\) & \(23\) dBm \\
A2A antenna gain & \(G_i,G_j\) & \(0\) dBi omnidirectional \\
LoS indicator & \(\ell_{ij,t}\) & LoS test over \(\EnvMap\) \\
Effective UAV antenna height & \(h_{\mathrm{eff}}\) & \(70\) m  ~\cite{etsi3gpp38901} \\
Receiver height & \(\diagup\) & \(1.5\) m~\cite{etsi3gpp38901} \\
A2G carrier frequency & \(\diagup\) & \(3.5\) GHz~\cite{mosoCanopy5G} \\
A2G RS EIRP & \(\RefSignalPower{i}\) & \(25\) dBm \\
\topdesc{A2G antenna\newline convention} & \topnote{\(\diagup\)} & Omnidirectional,\newline gain absorbed into EIRP \\
A2A nominal bandwidth & \(\diagup\) & \(20\) MHz~\cite{doodleLabsMeshRider} \\
A2A link threshold & \(\LinkThresh\) & \(-90\) dBm \\
\topdesc{Backhaul direct-access\newline threshold} & \topnote{\(\BackhaulThresh\)} & Calibrated to \(150\) m\newline horizontal radius \\
Backhaul gateway mode & \(\diagup\) & Master/remote/relay \\
\topdesc{Material reflection \newline coefficient Assigned\newline via height} & \topnote{\(\diagup\)} & Concrete: \(0.6\), \(\leq15\) m\newline Glass: \(0.4\), \(15\)--\(40\) m\newline Metal: \(0.85\),  \(>40\) m \\
Synchronization interval & \(\SyncPeriod\) & \(5\) steps \\
Number of UAVs & \(\NumUAV\) & \(10\) \\
Candidate trajectories & \(|\FeasProposalPool|\) & \(61\) \\
Local Observation dimension & \(\dim(\LocalObs_i)\) & \(224\) \\
Flocking and Learned steps & \(\WarmupSteps, \LearnSteps\) & \(15, 30\) steps \\
Coverage threshold & \(\CoverageThresh\) & \(-90\) dBm~\cite{3gpp38215} \\
SINR threshold & \(\SINRThresh\) & \(0\) dB \\
\topdesc{Balanced-score\newline weights} & \topnote{\(\MetricWeight\)} & \(0.25,0.10,0.20,0.20,\)\newline \(0.10,0.10,-0.05\) \\
Flocking scoring weights & \(\begin{gathered}\DispatchSepWeight,\DispatchCrowdWeight,\DispatchDegradeWeight\end{gathered}\) & \((0.25,0.6,0.3)\) \\
\bottomrule
\end{tabularx}
\end{table}

\begin{table}[t]
\centering
\caption{Training parameters.}
\label{tab:training_params}
\setlength{\tabcolsep}{2.0pt}
\renewcommand{\arraystretch}{1.05}
\begin{tabularx}{\columnwidth}{@{}>{\raggedright\arraybackslash}m{0.45\columnwidth}>{\centering\arraybackslash}m{0.2\columnwidth}>{\raggedright\arraybackslash}X@{}}
\toprule
Description & Notation & Value \\
\midrule
Candidate proposal per reduction & \(\NumCand\) & \(3\) \\
Relay-depth bound & \(\RelayDepth\) & \(2\) hops \\
Negotiation-round limit & \(\RoundLimit\) & \(3\) \\
JEPA continuity constants & \(\CoherenceTemp,\LatentEps\) & \((0.25,10^{-8})\) \\
Connectivity penalty & \(\ConnPenalty\) & \(1.0\) \\
PPO constants & \(\begin{gathered}\PPOClipEps,\PPOValueCoef,\PPOEntropyCoef\end{gathered}\) & \((0.2,0.5,0.01)\) \\
\bottomrule
\end{tabularx}
\vspace{-8pt}
\end{table}

\begin{figure*}[t]
\centering
\includegraphics[width=0.98\textwidth]{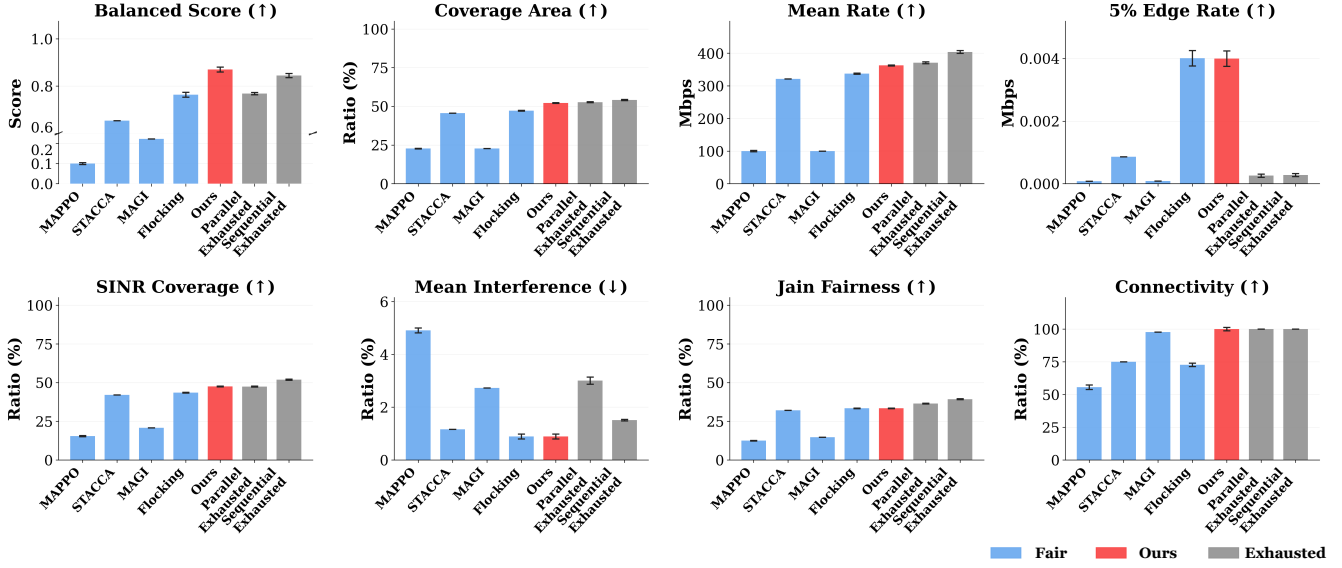}
\vspace{-8pt}
\caption{Closed-loop performance under the shared RadioDynamics evaluation.
Fair baselines are shown in \textcolor{blue}{blue}, WONDER in \textcolor{red}{red}, and Exhausted references in \textcolor{gray}{gray}.
The balanced score summarizes the normalized metric vector.}
\vspace{-8pt}
\label{fig:overall}
\end{figure*}

We adopt $62$ metropolitan scenes, such as Hong Kong, New York, and Tokyo,
for RadioDynamics construction, as illustrated in Fig. \ref{fig:real}.
The digital-twin pipeline first converts OpenStreetMap geometry into mesh-based city models and exports the scene assets to Sionna RT.
Sionna RT first assigns empirically selected concrete, glass,
and metal materials to the city models based on building height,
and then precomputes the ground-level RSRP field for each transmitter setting.
Each scene stores these transmitter settings together with the resulting dense RSRP fields.
The configuration is summarized in Table~\ref{tab:params}, while the training-related parameters are given in Table~\ref{tab:training_params}.

The $62$ scenes are further split into $42$ training scenes, $9$ validation scenes, and $11$ test scenes, with $4$ entries per scene.
Each method is trained with seed $0$ and evaluated with seeds \(\{0,\ldots,5\}\).
For testing, all methods first use Flocking for $\WarmupSteps =15$ steps to spread the swarm,
and the evaluated algorithm then runs for $\LearnSteps = 30$ steps from the same warm-up state.
At each step, the scene-aware A2A graph is updated from the UAV positions,
digital-twin line of sight,
and the link-budget feasibility rule in Section~\ref{sec:system_model}.
All experiments use the metrics defined in Section~\ref{sec:system_model},
namely RSRP coverage, 5th-percentile delivered rate, SINR coverage,
mean delivered rate, Jain fairness, backhaul connectivity, and mean interference.
The balanced score uses the weights listed in Table~\ref{tab:params}.

The deployment uses separate radios for A2G service broadcast,
A2A swarm relay,
and ground backhaul.
A temporary backhaul point is placed at the terminal gateway,
where a representative digital modem serves as the gateway and the same modem on UAVs serves as the remote or relay node. 
%
Each UAV carries an A2G service-radio payload for producing the RSRP footprint,
represented by a Canopy 5GID1 n78 small-cell radio envelope for the \(3.5\)-GHz,
\(100\)-MHz,
\(25\)-dBm-EIRP setting~\cite{mosoCanopy5G}.
Inter-UAV proposal exchange is carried by an A2A mesh radio,
represented by Mesh Rider Mini,
with SL5200 as a higher-end MANET alternative~\cite{doodleLabsMeshRider}.
A UAV beyond the \(150\)-m direct-root radius reaches the gateway through an A2A relay path to a gateway-attached UAV.

We select representative MARL and communication-bottleneck methods as the primary baselines.
MAPPO~\cite{yu2022surprising},
STACCA~\cite{sinha2025stacca},
and heuristic-based Flocking \cite{atincc2020swarm, wu2021multi, capelli2020connectivity, lin2021online} use the shared static scene prior and deployment-visible local observations without multi-hop communication.
MAGI~\cite{ding2024magi} and ours use the deployment-feasible communication interface,
including one-hop communication,
topology, and selected trajectories.
The ground-truth Reference Signal Received Power (RSRP) is not visible to deployable methods.
We also evaluate an Exhausted method under an oracle view with complete swarm and evaluator-side RSRP information.
Both variants are evaluated from the same $6$ seed-induced initial states: 
the Parallel variant lets all agents choose synchronously,
while its Sequential variant lets them choose one after another in a fixed order.

\subsection{Experiemental Results}
\subsubsection{Performance Superiority}
Fig.~\ref{fig:evolution} first illustrates a representative scene-level rollout,
linking the UAV motion,
RSRP field, A2A relay topology, and the bottom Score curve.
During the flock warm-up stage as in Fig. \ref{fig:evolution_start},
the UAVs spread from the initial cluster over the first 15 steps and form a broader relay structure over the city scene.
Afterward, the three rollout behaviors emerge:
Flocking keeps dispersing but encounters disconnection at $t=35$,
while STACCA fluctuates due to counterfactual mismatch.
WONDER instead achieves steady improvement,
adjusts the swarm layout while preserving relay links,
and continues increasing the score until the final step.
In this scene,
the WONDER trajectory ends with a $0.413$ score advantage over Flocking and $0.162$ over STACCA. 
We further evaluate all $7$ methods under the same closed-loop protocol and report averaged results in Fig.~\ref{fig:overall}.
Fig.~\ref{fig:overall} shows that WONDER attains the highest balanced score. 
The gain over heuristic-based Flocking comes from limiting excessive swarm dispersion, while the gain over STACCA comes from resolving conflicts among locally selected counterfactual decisions. 
Compared with Sequential Exhausted, WONDER further benefits from its negotiated decision order, which avoids the local optima caused by fixed-order decisions in building-dense scenarios, 
which is demonstrated in Fig.~\ref{fig:advantage_over_exhausted}. 

\begin{figure}[!tb]
\centering
\includegraphics[width=0.98\linewidth]{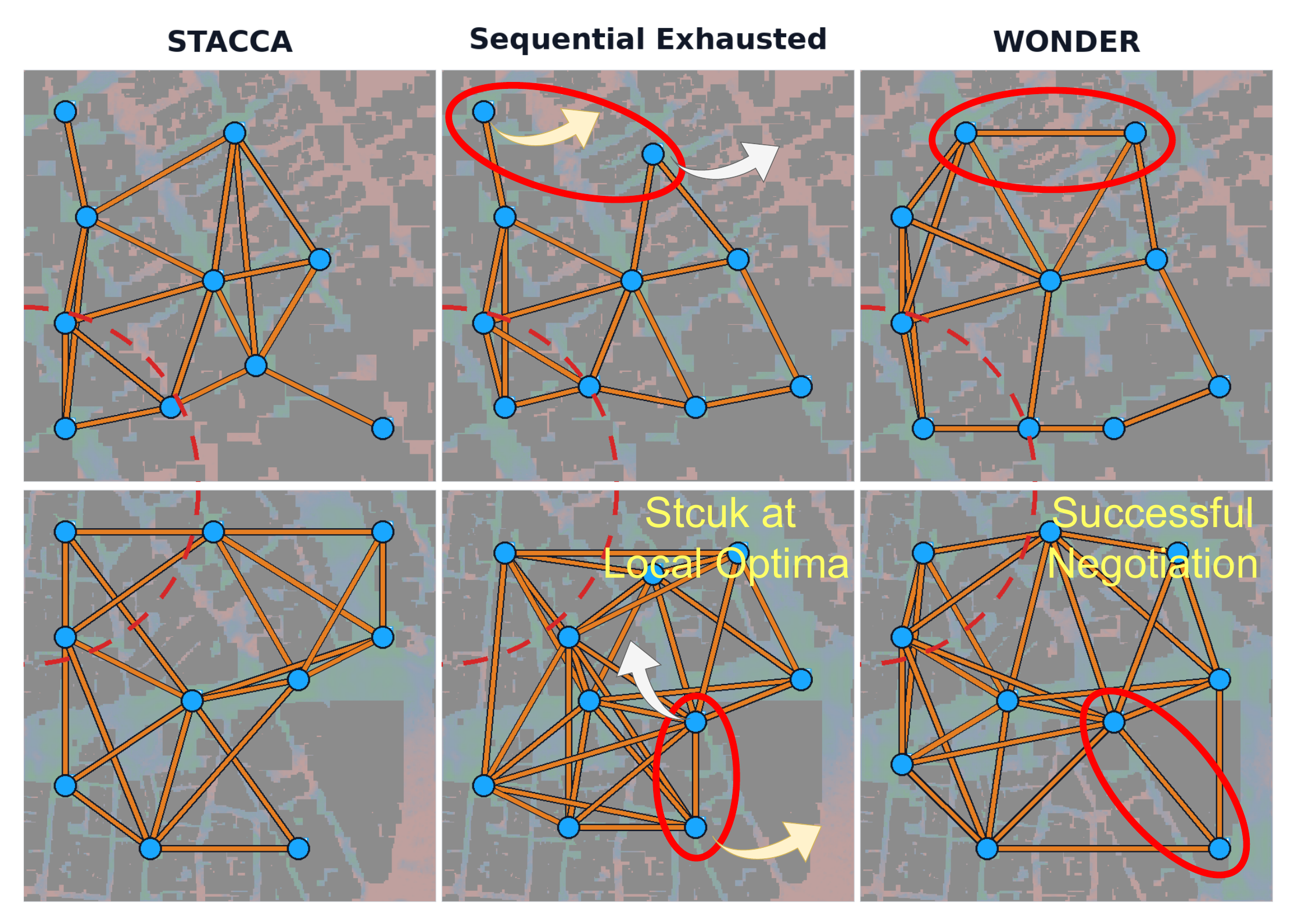}
\vspace{-12pt}
\caption{Representative cases where our method outperforms Sequential Exhausted.
Sequential Exhausted visits UAVs in a fixed order and greedily commits each visited UAV's locally best action. 
The \textcolor{gray}{gray} actions yield marginal gain, 
while the \textcolor{yellow!70!black}{yellow} actions would produce higher global gain if committed earlier. 
Our method prioritizes the \textcolor{yellow!70!black}{yellow} actions through proposal-level negotiation, avoiding the local optimum induced by fixed-order greedy selection.}
\label{fig:advantage_over_exhausted}
\vspace{-8pt}
\end{figure}

\subsubsection{Feasibility Study}

As discussed in Appendix \ref{app:comm}, Table~\ref{tab:comm} assesses whether the proposed communication mechanism can be carried by representative radios, by counting the payloads of local-vector exchange,
WONDER proposal aggregation,
and raw full-map transfer under the same topology.
Compared to STACCA, WONDER adds \(147.456\) kbit of proposal-aggregation payload per synchronization round,
resulting in a total of \(270.336\) kbit.
Using the \(80\) Mbps A2A entry \cite{doodleLabsMeshRider} as a nominal reference,
a \(1\)-s deployment step provides \(400\) Mbit over the \(5\)-step synchronization round.
The WONDER total therefore uses only \(0.0676\%\) of this nominal round capacity,
indicating that the communication load is well within the bandwidth of the listed devices while yielding superior coverage performance.

\begin{table}[t]
\caption{Payload upper bounds for message exchange.}
\label{tab:comm}
\centering
\setlength{\tabcolsep}{1.5pt}
\renewcommand{\arraystretch}{1.06}
\begin{tabular}{@{}>{\raggedright\arraybackslash}m{0.24\linewidth}>{\raggedright\arraybackslash}m{0.27\linewidth}>{\centering\arraybackslash}m{0.18\linewidth}>{\centering\arraybackslash}m{0.20\linewidth}@{}}
\toprule
Entry & Scope & No. of Messages & \shortstack[c]{Payload\\(kbit)} \\
\midrule
STACCA~\cite{sinha2025stacca} & Local vectors & \(120\) & \(122.88\) \\
MAGI~\cite{ding2024magi} & Neighborhood Message & \(120\) & \(245.76\) \\
WONDER & Local + proposal & \(264\) & \(270.336\) \\
Full Connection & Full Radio Map (raw) & \(120\) & \(470400\) \\
\bottomrule
\end{tabular}
\vspace{-8pt}
\end{table}

\subsubsection{Ablation Analysis}
Fig.~\ref{fig:Stage1_Training} further confirms that the JEPA Radio World Model and the Direct Ranking (DR) variant,
which uses the same backbone network and removes \(\AlignLoss\),
both converge during local proposal training.
Fig.~\ref{fig:Stage1_Ensemble_Recall} then evaluates
%
%
whether the union of locally retained Top-$\NumCand$ candidates contains the globally counterfactual best action over the \(3\) negotiation rounds.
At $\NumCand = 3$,
JEPA achieves \(99.62\%\) recall,
outperforming DR by \(4.92\%\),
%
demonstrating that local Top-$3$ retention is sufficient for reliable subsequent election.

\begin{figure}[tb]
\centering
\vspace{-4pt}
\includegraphics[width=0.98\linewidth]{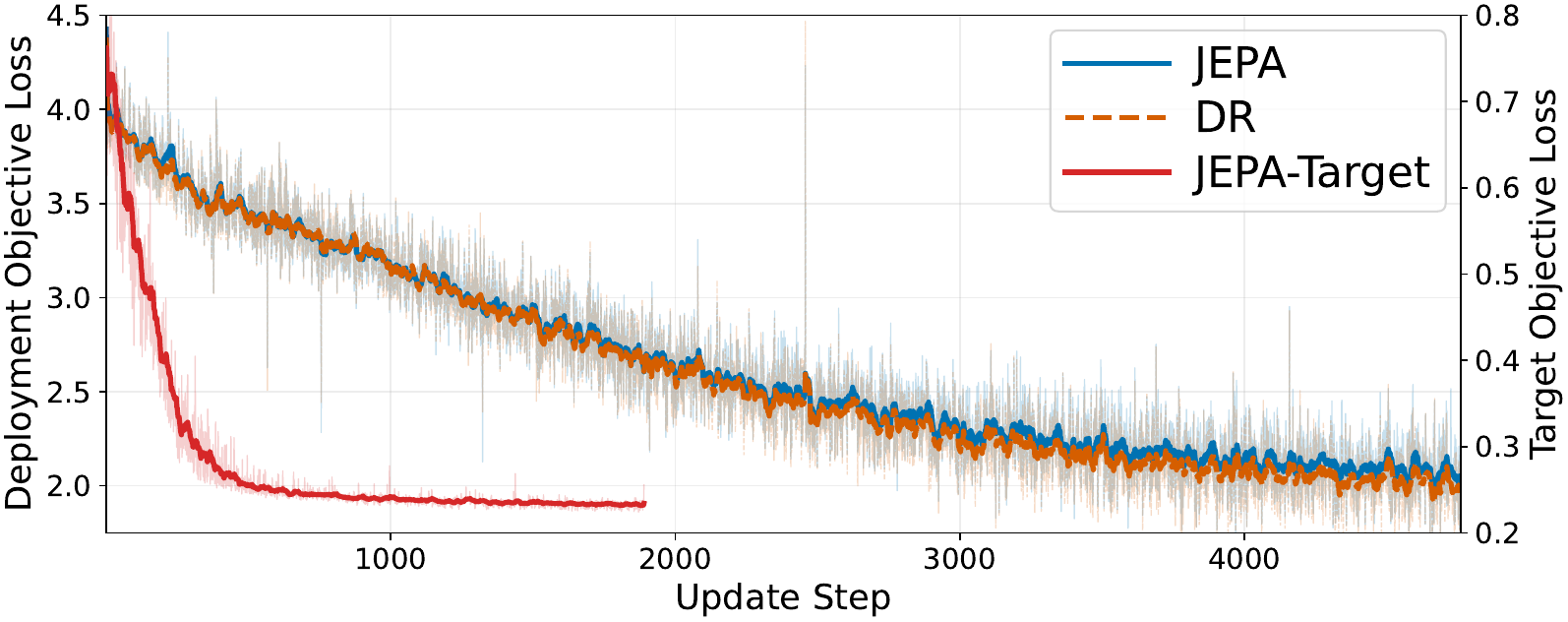}
\vspace{-6pt}
\caption{Training curves for JEPA and DR local proposal models.}
\vspace{-15pt}
\label{fig:Stage1_Training}
\end{figure}

\begin{figure}[tb]
\centering
\vspace{-2pt}
\includegraphics[width=0.92\linewidth]{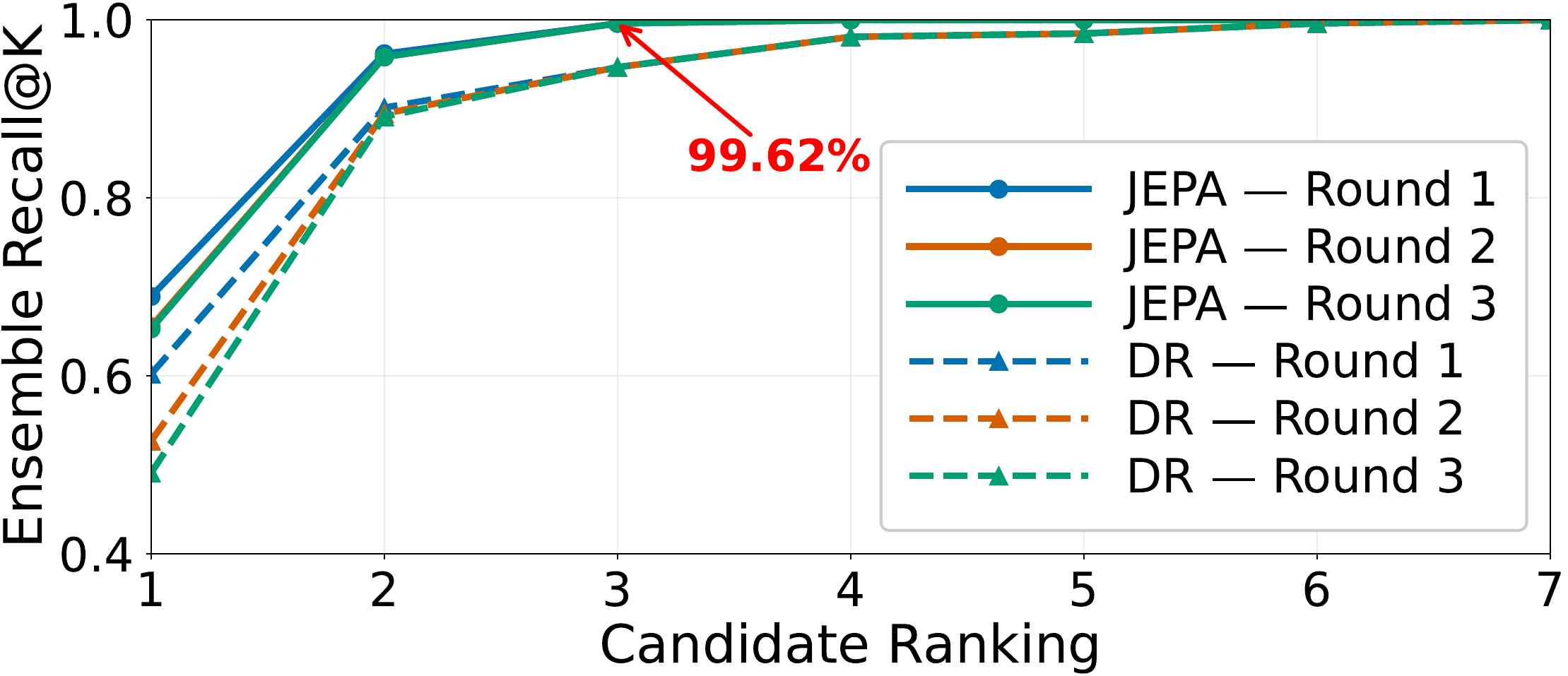}
\vspace{-6pt}
\caption{Offline recall of retained proposals for JEPA and DR.}
\label{fig:Stage1_Ensemble_Recall}
\vspace{-8pt}
\end{figure}

\begin{figure}[tb]
\centering
\includegraphics[width=0.98\linewidth]{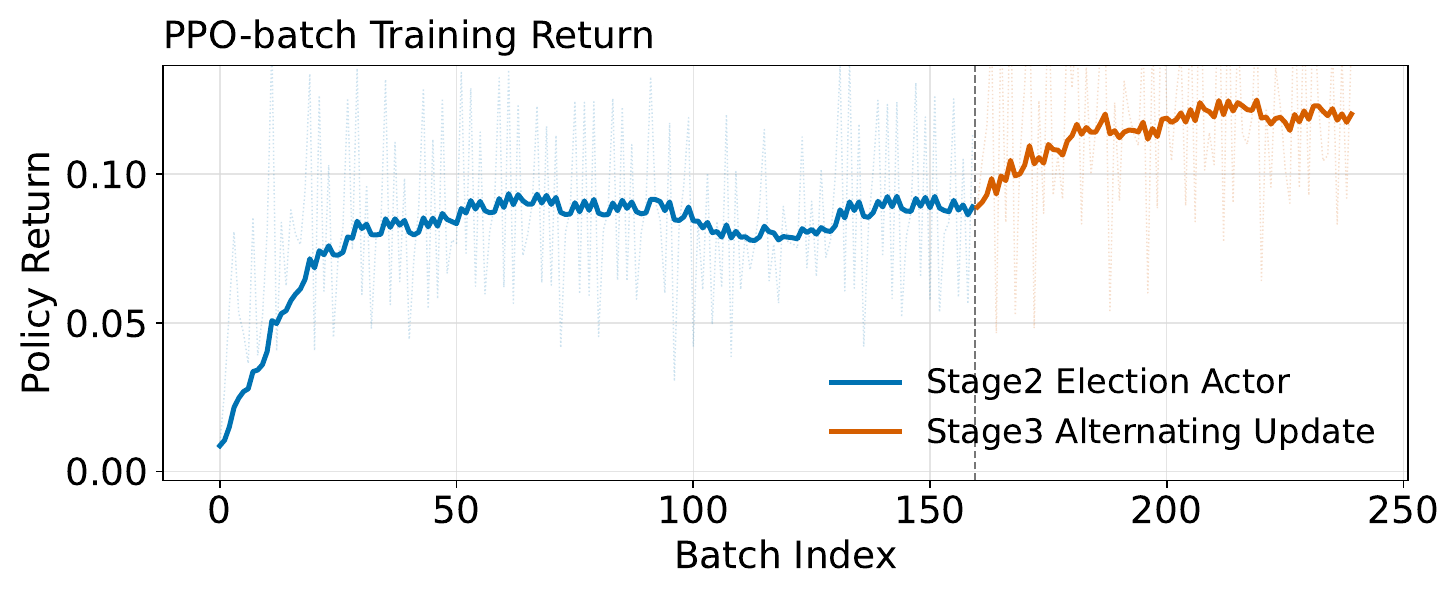}
\vspace{-2pt}
\caption{Training curve of the election policy.}
\label{fig:Stage23_Curve}
\vspace{-15pt}
\end{figure}

\begin{table*}[tb]
\centering
\caption{Ablation analysis of the main WONDER components under the shared evaluation setting.}
\label{tab:ablation_results}
\scriptsize
\setlength{\tabcolsep}{2.2pt}
\renewcommand{\arraystretch}{0.88}
\resizebox{\textwidth}{!}{%
\begin{tabular}{@{}*{10}{>{\centering\arraybackslash}c}@{}}
\toprule
Method
& \shortstack{Coverage\\Area (\%)}
& \shortstack{Mean Rate\\(Mbps)}
& \shortstack{5\% Edge\\Rate (Mbps)}
& \shortstack{SINR\\Coverage (\%)}
& \shortstack{Mean\\Interference (\%)}
& \shortstack{Jain\\Fairness (\%)}
& \shortstack{Connectivity\\(\%)}
& \shortstack{Balanced\\Score} \\
\midrule
w/o JEPA Encoder & 43.1 & 300.8 & 0.00357 & 38.5 & 0.98 & 29.8 & 66.5 & 0.659 \\
w/o Election Module & 44.5 & 322.2 & 0.00381 & 41.0 & 0.94 & 31.7 & 68.4 & 0.708 \\
\underline{w/o Alternating update} & \underline{49.4} & \underline{353.2} & \underline{0.00420} & \underline{45.0} & \underline{0.85} & \underline{35.0} & \underline{75.7} & \underline{0.808} \\
\textit{w/o Multi-Round Negotiation} & \textit{48.3} & \textit{343.3} & \textit{0.00408} & \textit{44.5} & \textit{0.87} & \textit{34.1} & \textit{74.0} & \textit{0.784} \\
\textbf{Our full method} & \textbf{52.2} & \textbf{362.5} & \textbf{0.00400} & \textbf{47.5} & \textbf{0.89} & \textbf{33.4} & \textbf{100.0} & \textbf{0.870} \\
\bottomrule
\end{tabular}%
}
\vspace{-6pt}
\end{table*}

Table~\ref{tab:ablation_results} reports the main ablation results for WONDER.
It compares the full method with four variants: replacing JEPA with DR,
replacing the Election Module with deterministic reduction,
applying JEPA encoding only in the first negotiation round without subsequent representation updates,
and removing the alternating update procedure.
The full method achieves the highest balanced score,
\(0.870\),
which is \(7.67\%\) higher than the strongest ablation,
the variant without alternating updates,
%
%
and is the only variant that reaches \(100\%\) connectivity.
Across the ablations,
removing JEPA,
learned election,
alternating optimization,
or multi-round negotiation
reduces the balanced score to \(0.659\), \(0.708\), \(0.808\), and \(0.784\), respectively,
indicating that each component is associated with a measurable loss in closed-loop performance when omitted or replaced.
Fig.~\ref{fig:Stage23_Curve} provides complementary evidence through the later-stage training curve of the election policy.
The trend after alternating refinement is consistent with the table-level gain of the full method,
suggesting that alternating optimization helps WONDER balance coverage,
rate,
connectivity,
and interference in the evaluated closed-loop setting.

\section{Conclusion}
\label{sec:conclusion}

This paper has studied deployable multi-UAV coverage restoration under the joint difficulty of hidden radio consequences and sparse inter-UAV communication.
To address the resulting radio observations and local-global coordination asymmetries,
WONDER has introduced a JEPA-based radio world model to predict candidate trajectories' incremental radio consequences under the evolving negotiation context,
while multi-round negotiation benefits the PPO-based election actor to generate superior policy.
The empirical analysis identifies the counterfactual composition gap caused by the synchronized execution of individually evaluated proposals, 
and the theoretical analysis then shows that recursively sequential commitment admits an advantage over the parallel counterfactual plan, 
which applies to WONDER up to bounded election regret when proposal coverage holds. 

Moreover, we have built RadioDynamics, a comprehensive simulation environment that integrates UAV mobility,
radio propagation, inter-UAV communication modeling,
and digital-twin geometry with ray-traced fields in $62$ metropolitan scenes.
On $11$ testing scenes of RadioDynamics, WONDER has achieved the highest balanced score among the evaluated methods ($0.870$) and reached $100\%$ gateway connectivity, validating its superiority.
\appendices
\section{Communication Payload Calculation}
\label{app:comm}
\subsection{Communication Accounting for Local Vectors and WONDER Proposals}

We derive the payload upper bounds summarized in Tab.~\ref{tab:comm} by counting the transmitted entries in each local vector and each WONDER message. 
We use a representative audited topology that partitions the swarm into subgraphs of sizes
\(3\), \(4\), and \(3\), which gives \(\sum_{q\in\{3,4,3\}}q(q-1)=24\) directed transmissions per local exchange.
A synchronization round therefore has the shared local-aggregation factor \(5\times24\) for
methods that aggregate over the local topology.
STACCA and WONDER local aggregation are counted with the \(32\)-float component size in
Table~\ref{tab:params}, MAGI uses \(64\)-float messages, and Full Connection transfers a \(350\times350\) RSRP map.
WONDER message aggregation is counted separately with \(\RoundLimit=3\), \(\NumCand=3\), and \(\RelayDepth=2\) from Table~\ref{tab:training_params}, giving up to
\(3\times[(10\times3)+(3\times3\times2)]\times32\) float32 values.
Thus, STACCA, MAGI, WONDER, and Full Connection contain \(5\times24\times32\), \(5\times24\times64\),
\(3\times[(10\times3)+(3\times3\times2)]\times32\), and \(5\times24\times(350\times350)\) float32 values,
respectively. Multiplying these counts by \(32\) bits gives the table payloads.
Full Connection requires \(1.74\times10^3\) times the WONDER total payload.

\section{Detailed Loss Function}
\label{app:loss}

\subsection{JEPA Radio World Model Objectives}
\label{app:jepa_loss}

%
We use the rollout-induced target score $\ProposalScore_{i,k}^{r}$ defined in \eqref{eq:target_metric_change}.
Let $k^*$ be denote the best trajectory index for UAV $i$ in round $r$, 
the auxiliary ranking loss is
\begin{equation}
\RankLoss
=
-\mathbb E_{i,r}
\left[
\log \frac{
\exp\!\left(\ProposalScore_{i,k^*}^{r}/\sigma_{\mathrm{score}}\right)
}{
\sum_{k':\Traj_{i,k'}\in\FeasProposalPool_r}
\exp\!\left(\ProposalScore_{i,k'}^{r}/\sigma_{\mathrm{score}}\right)
}
\right].\label{eq:ranking_loss}
\end{equation}
Here, the expectation is over active UAV-round pairs, and $\sigma_{\mathrm{score}}>0$ is the fixed score-normalization constant
used for the trajectory-score head.

\subsection{Election Module Objective}
\label{app:election_loss}

This subsection gives the detailed objective used to train the election module $\ElectionModule$ after the JEPA Radio World Model
is fixed.
At synchronization step $t$ and election round $r$, the deterministic reductions produce the global trajectory pool
$\GlobalProposalPool_r$, and $\ElectionModule$ selects either one trajectory from this pool or ends the negotiation.
Let $\JointPlan_t^{r}$ denote the partial joint plan before round $r$.
If a trajectory is selected, let $\CommittedTraj_r$ be the trajectory contained in that trajectory; the post-election plan is
$\JointPlan_t^{r+1}=\JointPlan_t^{r}\oplus\CommittedTraj_r$.
The reward is defined by balanced-objective improvement with a penalty for gateway-connectivity violation:
\begin{equation}
\begin{aligned}
\ElectionReward
&=
\MetricWeight^{\top}
\!\left[
\MetricVec(\SwarmState_{t+\SyncPeriod};e)\Big|_{\JointPlan_t^{r+1}}
-
\MetricVec(\SwarmState_{t+\SyncPeriod};e)\Big|_{\JointPlan_t^{r}}
\right]
\nonumber\\
&\quad -
\ConnPenalty
\!\left[
1-\Connectivity(\SwarmState_{t+\SyncPeriod};e)\Big|_{\JointPlan_t^{r+1}}
\right].
\end{aligned}
\label{eq:app_election_reward}
\end{equation}
where the restriction $\big|_{\JointPlan_t^{r}}$ means that the terminal metric vector is evaluated after executing $\JointPlan_t^{r}$ over the synchronization horizon. The connectivity term penalizes violation of the hard gateway-connectivity requirement in \eqref{eq:constrained_deployment_problem}, with weight $\ConnPenalty$.
If $\ElectionModule$ selects to end negotiation, the transition reward is set to zero.
The election module is optimized with the PPO objective
\begin{equation}
\ElectionUpdateLoss
=
\PPOClipLoss
+
\PPOValueCoef\PPOValueLoss
-
\PPOEntropyCoef\PPOEntropyBonus .
\label{eq:app_election_objective}
\end{equation}
Here, $\PPOValueCoef$ and $\PPOEntropyCoef$ weight the value-loss and entropy terms, respectively.
The clipped policy loss is
\begin{equation}
\begin{aligned}
\PPOClipLoss
&=
-\mathbb E_{r}
\!\left[
\min\!\left(
\begin{aligned}
&\PPOProbRatio_{r}(\ElectionParams)\PPOAdvantage_{r},\\
&\operatorname{clip}\!\left(
\PPOProbRatio_{r}(\ElectionParams),
1-\PPOClipEps,
1+\PPOClipEps
\right)
\PPOAdvantage_{r}
\end{aligned}
\right)
\right],
\end{aligned}
\label{eq:app_ppo_clip_loss}
\end{equation}
where $\PPOProbRatio_{r}(\ElectionParams)$ is the PPO probability ratio between the current election policy and the behavior policy
for the selected election action at round $r$, and $\PPOAdvantage_{r}$ is the corresponding advantage estimate under
\eqref{eq:app_election_reward}. The clipping radius is $\PPOClipEps$.
The value loss trains the centralized critic with 
$
\PPOValueLoss
=
\mathbb E_{r}
\!\left[
\HuberLoss\!\left(
\PPOValue_{r},
\PPOReturn_{r}
\right)
\right],
$
%
where $\PPOValue_{r}$ is the centralized critic estimate and $\PPOReturn_{r}$ is the rollout return.
The entropy term $\PPOEntropyBonus$ is computed over the election action set formed by $\GlobalProposalPool_r$.

\bibliographystyle{IEEEtran}
\bibliography{reference}
\end{document}